\documentclass[runningheads]{llncs}
\usepackage{microtype}
\usepackage{todonotes}
\usepackage{amsmath}
\usepackage{amsmath,amsfonts,amssymb} 
\usepackage{url}
\usepackage{stmaryrd}
\usepackage{hyperref}
\usepackage{enumerate}
\usepackage{todonotes}
\usepackage{wrapfig}
\usepackage{pgfplots}
\usepackage{multirow}
\usepackage{verbatim}
\usepackage{mathtools}
\usepackage{booktabs}

\usepackage{enumerate}
\usepackage{paralist}
\usepackage{extarrows}
\usepackage{mathtools}

\usepackage{wasysym}

\usepackage{algorithm}
\usepackage{algorithmic}
\usepackage{listings}
\usepackage[title]{appendix}
\usepackage{marginnote}
\usepackage{listings}
\usepackage{url}
\usepackage{wrapfig}
\usepackage{graphicx}
\usepackage{tikz}
\usepackage{xcolor,colortbl}
\usetikzlibrary{positioning,shadows.blur}
\usetikzlibrary{arrows,automata}
\usepackage{pifont}
\usepackage{hhline}
\allowdisplaybreaks

\newcommand{\M}{\mathbin{\mathbf{M}}} 
\def\Masking{\preceq_{m}}

\def\FMask{f_{m}}

\def\Refuter{\mathsf{R}}
\def\Verifier{\mathsf{V}}
\def\Probabilistic{\mathsf{P}}
\def\ErrorSt{v_{\text{err}}}
\def\ErrorStSG{\ErrorSt^{\SymbG_{A,A'}}}
\def\InitVertex{v_0^\mathcal{G}}
\def\InitVertexSG{v_0^{\SymbG_{A,A'}}}

\def\SigmaF{\Sigma_{\mathcal{F}}}

\def\val{\mathop{\textup{val}}}
\def\support#1{\mathit{Supp}(#1)}
\def\image#1{\mathit{Im}(#1)}
\def\faults{\mathcal{F}}
\def\RFP{\mathit{RFP}}
\def\Prob#1#2{\mathit{Prob}^{#1, #2}}
\def\couplings#1#2{\mathbb{C}(#1,#2)}
\def\vertices#1{\mathbb{V}(#1)}
\def\pr#1#2{{#2}[{#1}]}

\DeclareMathOperator*{\argmax}{argmax}

\def\RelCoupling{R^{\#}}
\newcommand{\MaskCoup}{\mathbin{\mathbf{M^{\#}}}} 
\def\Expect#1#2{\mathbb{E}^{#1, #2}}
\def\SymbG{\mathcal{SG}}
\def\StochG{\mathcal{G}}
\def\Stochgame#1{\mathcal{#1}}
\def\MilestoneG{\mathcal{MG}}
\def\post{\mathit{Post}}
\def\pre{\mathit{Pre}}

\def\AFairpre{\forall\mathit{Pre}_f}
\def\EFairpre{\exists\mathit{Pre}_f}
\def\SymbAFairpre{\forall\mathit{Pre}^S_f}
\def\SymbEFairpre{\exists\mathit{Pre}^S_f}
\def\Eq{\mathit{Eq}}
\def\out{\mathit{out}}
\def\Dist{\mathcal{D}}
\def\Dirac{\Delta}
\def\reward{\mathit{r}}
\def\LP{\mathit{LP}}
\def\tuple#1{{\langle{#1}\rangle}}

\def\strat#1{\pi_{#1}}

\def\starredstrat#1{\pi^*_{#1}}
\def\Memoryless{\mathsf{M}}
\def\Deterministic{\mathsf{D}}

\def\upperbound{\mathsf{u}}

\def\Strategies#1{\Pi_{#1}}

\mathchardef\mhyphen="2D

\newcommand*{\xrightarrowprime}[2][]{\mathrel{{\xrightarrow[#1]{#2}}{}'}}
\def\PRISM{\textsf{PRISM}}
\def\LTL{\textsf{LTL}}
\def\Bellman{L}

\renewenvironment{proof}{\noindent\textbf{Proof.}}{\hfill$\diamond$}

\usepackage{color}

\definecolor{lightblue}{RGB}{231,255,255}
\definecolor{lightred}{RGB}{255,224,224}
\definecolor{lightgreen}{RGB}{224,255,224}
\definecolor{lightyellow}{RGB}{255,255,224}
\definecolor{lightpurple}{RGB}{255,224,255}
\definecolor{darkerred}{RGB}{64,0,0}
\definecolor{darkred}{RGB}{128,0,0}
\definecolor{darkblue}{RGB}{0,0,128}
\definecolor{darkgreen}{RGB}{0,128,0}
\definecolor{darkpurple}{RGB}{128,0,128}
\definecolor{black}{RGB}{0,0,0}

\newcommand{\colorpar}[3]{\colorbox{#1}{\parbox{#2}{#3}}}
\newcommand{\marginremark}[3]{\marginpar{\colorpar{#2}{\linewidth}{\color{#1}#3}}}
\newcommand{\remarkPC}[1]{\marginremark{darkblue}{lightblue}{\scriptsize{[PC]~ #1}}}

\newcommand{\highlightedremark}[4]{%
  \begin{center}\fcolorbox{#1}{#2}{%
  \begin{minipage}{.98\linewidth}\color{#1}%
  \textbf{\THICKhrulefill[ #3 ]\THICKhrulefill}%
  \par\noindent#4\end{minipage}}\end{center}%
}

\makeatletter
\def\THICKhrulefill{\leavevmode \leaders \hrule height 5pt\hfill \kern \z@}
\makeatother

\newcommand{\hrmkRD}[1]{\highlightedremark{darkgreen}{lightgreen}{RD}{#1}}

\graphicspath{{./Figs/}}

\title{A Stochastic Game Approach to Masking Fault-Tolerance: Bisimulation and Quantification} 

\titlerunning{A Stochastic Game Approach to Masking Fault-Tolerance}

\author{Pablo F. Castro \inst{1,3} \and
Pedro R. D'Argenio \inst{2,3,4} \and \\
Luciano Putruele \inst{1,3}  \and 
Ramiro Demasi \inst{2,3}
}

\authorrunning{P.F. Castro et al.}
\institute{Departamento de Computaci\'on, FCEFQyN, Universidad Nacional de 
  R\'{\i}o Cuarto,
  Argentina
  \and 
  FAMAF, Universidad Nacional de C\'ordoba
  Argentina,
  \and
  Consejo Nacional de Investigaciones Cient\'ificas y T\'ecnicas (CONICET), Argentina \and
  Saarland University, Saarland Informatics Campus,
  Germany}

\begin{document}

\maketitle

\begin{abstract}
We introduce a formal notion of masking fault-tolerance  between
probabilistic transition systems based on a variant of probabilistic
bisimulation (named masking simulation).  We also provide the corresponding probabilistic game
characterization.  Even though these games could be infinite, we propose
a symbolic way of representing them, such that it can be decided in polynomial time 
if there is a masking simulation between two probabilistic transition systems.
%
We use this notion of masking  to quantify the level of masking fault-tolerance 
exhibited by almost-sure failing systems, i.e., those systems that eventually fail with probability $1$.  
The level of masking fault-tolerance of almost-sure failing systems can be calculated by solving a 
collection of functional equations.
We produce this metric in a setting in which the minimizing player behaves in a strong fair way (mimicking the idea of fair environments), and limit our study to memoryless strategies due to the infinite nature of the game.
We implemented these ideas in a prototype tool,  and performed an experimental evaluation.
\end{abstract}

\setcounter{page}{1}

\section{Introduction} \label{sec:intro}

Fault-tolerance is one important characteristic of modern software.  This is particularly true for critical software like banking software, 
automotive applications, communication protocols,  avionics software, etc.
However, in practice, it is hard to quantify the level of fault-tolerance provided by computing systems.  In most cases fault-tolerant systems are built using  ad-hoc techniques which are 
based on experience and, many times, lack a mathematical foundation.
Furthermore, faults usually have a probabilistic nature.  Therefore, concepts coming from probability theory become necessary when developing fault-tolerant software.  In this paper we
provide a framework aimed at analysing the fault-tolerance exhibited by concurrent probabilistic systems.  This encompasses  the probability of occurrence of faults as well as the use
of randomized algorithms for improving the fault-tolerance of systems.
	
In practice, there are different types of fault-tolerance,  \emph{masking fault-tolerance} (when both the safety and liveness properties are preserved under the occurrence of faults), 
\emph{non-masking fault-tolerance} (when only liveness properties are preserved) and \emph{failsafe fault-tolerance} (when only safety properties are preserved). 
Among them, masking fault-tolerance is often acknowledged as the most desirable kind of fault-tolerance, because all the properties of the nominal (i.e., non-faulty) system are preserved under faulty behavior. 
However, in many settings, requiring full masking fault-tolerance is unrealistic.  In particular, for those systems that are not designed to terminate and the degradation of hardware, or software, components will
 eventually lead to a failure  (i.e., a behavior that deviates from the expected system's behavior). One of the main applications of the framework described in the forthcoming sections is the quantification of the 
 amount of masking fault-tolerance provided by systems before they enter into a failure.  This measure provides a tool for selecting a fault tolerance mechanism, or for balancing multiple mechanisms (e.g.,  to which extent it is worth the cost of efficient hardware redundancy against time demanding software artifacts).
	
	During the last decade, significant progress has been made towards 
defining suitable metrics or distances for diverse types of quantitative 
models including real-time systems \cite{HenzingerMP05}, probabilistic 
models \cite{Bacci0LM17,BacciBLMTB19,DesharnaisGJP04,DesharnaisLT11,TangB18}, 
and metrics for linear and branching systems \cite{CernyHR12,AlfaroFS09,Henzinger13,LarsenFT11,ThraneFL10}. 
Some authors have already pointed out that these metrics can be useful to reason 
about the robustness and correctness of a system, notions related to fault-tolerance. 
In \cite{CastroDDP18b}, we presented a notion of masking fault-tolerance 
between systems built on a  simulation relation and a corresponding 
game representation with quantitative objectives. In this paper we review these ideas in a probabilistic setting, and define a probabilistic version of 
this characterization of masking fault-tolerance.

More specifically,  we start characterizing probabilistic masking fault-tolerance via a variant of probabilistic bisimulation.  This \emph{masking simulation} relates two probabilistic transition systems. The first one acts as a specification  of the intended behavior (i.e., nominal model) and the 
second one as the fault-tolerant implementation (i.e., the extended model with  faulty behavior). The existence of a masking simulation implies 
that the implementation masks the faults. This simulation relation can be captured as a stochastic game played by a 
\emph{Verifier} and a \emph{Refuter}. If the \emph{Verifier} wins, there is a probabilistic masking simulation.
If instead the \emph{Refuter} wins, the implementation is not masking fault-tolerant.
	 These games rely on the notion of  couplings between probabilistic distributions and, as a consequence, the numbers of vertices
of the game graphs is infinite. To tackle this problem, we introduce a symbolic representation of these games 
where the couplings are symbolically captured by means of equation systems.
The size of this symbolic graph 
is polynomial on the size of the input systems. Moreover, the simulation games
can be solved via their symbolic representation.

In practice masking fault-tolerance comes in a 
quantitative fashion, and thus we enrich the games with quantitative objectives.  This makes it possible to quantify the 
amount of masking tolerance provided by the implementations.
We focus on games that almost-surely fail when the Refuter plays
fairly (i.e., those systems that eventually fail with probability
$1$).  Due to the infinite nature of the games, we restrict this
result to randomized memoryless strategies and show that the game is
determined under these conditions.
We show that the problem of deciding if the game is almost surely failing under fairness is polynomial.  Moreover, the value of the game can be computed 
by solving a collection of functional equations via a \emph{Value Iteration} 
algorithm \cite{ChatterjeeH08,Condon90,Condon92,KelmendiKKW18}.  We take such a value as the measure of fault-tolerance.


Summarizing our contribution,
\begin{inparaenum}[(1)]
\item%
  we define a notion of probabilistic masking simulation,
\item%
  provide a game characterization for it, and
\item%
  show that it can be decided in polynomial time
  (Sec.~\ref{sec:prob_mask_dist}).
Moreover, in Sec~\ref{sec:prob_almost_sure},
\item%
  we define an extension of the game with rewards and provide
  a payoff function that counts the number of ``milestones'' achieved
  by the implementation;
\item%
  we show that the game is determined provided it is almost-surely
  failing under fairness and memoryless strategies, and
\item%
  provide an algorithm to calculate it.
In addition,
\item%
  we provide a polynomial time algorithm to decide if a game
  is almost-surely failing under fairness.
\item%
  We finally present an experimental evaluation on some well-known
  case studies (Sec.~\ref{sec:experimental_eval}).
\end{inparaenum}
Full details and proofs can be found in the Appendix.

\section{Preliminaries} \label{sec:background}

We first introduce some basic definitions that will be necessary across the paper.

A (discrete) \emph{probability distribution} $\mu$ over a denumerable 
set $S$ is a function $\mu: S \rightarrow [0, 1] $  such that 
$\mu(S) \triangleq \sum_{s \in S} \mu(s) = 1$. 
Let $\Dist(S)$ denote the set of all probability distributions on $S$. $\Dirac_s \in \Dist(S)$ denotes the Dirac distribution for $s$, i.e., 
$\Dirac_s(s) =1$ and $\Dirac_s(s') = 0$ whenever $s'\neq s$.
The \textit{support} set of $\mu$ is defined by $\support{\mu} = \{s |~\mu (s) > 0\}$.

A \emph{Probabilistic Transition System} (PTS) is a structure $A = \langle S, \Sigma, \rightarrow, s_0 \rangle$ 
where 
\begin{inparaenum}[(i)]
\item%
  $S$ is a denumerable set of \emph{states} containing the
  \emph{initial state} $s_0 \in S$,
\item%
  $\Sigma$ is a set of \emph{actions}, and
\item%
  ${\rightarrow} \subseteq S \times \Sigma \times \Dist(S)$ is the
  \emph{(probabilistic) transition relation}.
\end{inparaenum}
We assume that there is always some transition leaving
from every state.


A distribution $w \in \Dist(S \times S')$ is a \textit{coupling} for $(\mu, \mu')$, with 
$\mu \in \Dist(S)$ and $\mu' \in \Dist(S')$, if $w(S, \cdot) = \mu'$ and 
$w(\cdot, S') = \mu$. $\couplings{\mu}{\mu'}$ denotes the set of all couplings for $(\mu, \mu')$.  It is worth
noting that this defines a polytope. $\vertices{\couplings{\mu}{\mu'}}$ denotes the set of all vertices of the corresponding polytope.
This set is finite if $S$ and $S'$ are finite.
For $R \subseteq S \times S'$, we say that a coupling $w$ for $(\mu, \mu')$ 
\textit{respects} $R$ if $\support{w} \subseteq R$ (i.e., $w(s, s') > 0 \Rightarrow s~R~s'$).
We define $\RelCoupling \subseteq \Dist(S) \times \Dist(S')$ by $\mu~\RelCoupling~\mu'$ if and only if
there is an $R$-respecting coupling for $(\mu, \mu')$.

A \emph{stochastic game graph} \cite{ChatterjeeH12} is a tuple $\mathcal{G} = ( V, E, V_1, V_2, V_\Probabilistic, v_0, \delta  ) $, where $V$ is a set of vertices with $V_1, V_2, V_\Probabilistic \subseteq V$ being a partition of $V$, $v_0\in V$ is the initial vertice,  $E\subseteq V \times V$,
and $\delta : V_\Probabilistic  \rightarrow \Dist(V)$ is a probabilistic transition function such that,  for all $v \in V_\Probabilistic$ and $v' \in V$:  $(v,v') \in E$ iff $v' \in \support{\delta(v)}$. 
If $V_\Probabilistic = \emptyset$, then $\mathcal{G}$ is called a $2$-player game graph. Moreover, if $V_1 = \emptyset$ or $V_2 = \emptyset$, then $\mathcal{G}$ is a \emph{Markov Decision Process} (or MDP). Finally, in case that $V_1= \emptyset$ and $V_2 = \emptyset$, $\mathcal{G}$ is a \emph{Markov chain} (or MC). For all states $v \in V$ we define $\post(v) = \{v' \in V \mid (v,v') \in E\}$, the set of successors of $v$. Similarly, $\pre(v') = \{v \in V \mid (v,v') \in E \}$ as the set of predecessors of $v'$. 
We assume that $\post(s) \neq \emptyset$ for every $v \in V_1 \cup V_2$. 

Given a game as defined above, a \emph{play} is an infinite sequence $\omega =  \omega_0, \omega_1, \dots$ such that $(\omega_k, \omega_{k+1}) \in E$ for every $k \in \mathbb{N}$.  The set of all plays is denoted by $\Omega$,
and the set of plays starting at vertex $v$ is written $\Omega_v$. A \emph{strategy} (or policy) for Player $i\in\{1,2\}$ is a function $\strat{i}: V^*  \cdot V_i \rightarrow \Dist(V)$ that assigns a probabilistic distribution to each finite sequence of states such that $\support{\strat{i}(w\cdot v)}\subseteq\post(v)$ for all $w\in V^*$ and $v\in V_i$.   The set of all the strategies for Player $i$ is named $\Strategies{i}$. A strategy $\strat{i}$ is said to be  \emph{pure} if, for every $w\in V^*$ and $v\in V_i$, $\strat{i}(w \cdot v)$  is a Dirac distribution, and it is called \emph{memoryless} if $\strat{i}(w  \cdot v) = \strat{i}(v)$, for every $w \in V^*$ and $v\in V_i$. 
Given two strategies $\strat{1} \in \Strategies{1}$, $\strat{2} \in \Strategies{2}$ and a starting state $v$,  the \emph{result} of the game  is a Markov chain, denoted by $\mathcal{G}^{\strat{1}, \strat{2}}_v$. 
As any Markov chain, $\mathcal{G}^{\strat{1}, \strat{2}}_v$ defines a
probability measure
$\Prob{\strat{1}}{\strat{2}}_{\Stochgame{G},v}$. 
An event $\mathcal{A}$ is a measurable set in the Borel $\sigma$-algebra generated by the cones of $\Omega$.
Thus, $\Prob{\strat{1}}{\strat{2}}_{\Stochgame{G},v}(\mathcal{A})$ is
the probability that strategies $\strat{1}$ and $\strat{2}$ generate a
play belonging to $\mathcal{A}$ from state $v$.
It would normally be convenient to use {\LTL} notation to define events. For instance,
$\Diamond V' =  \{ \omega = \omega_0,\omega_1,\dots \in \Omega \mid  \exists i : \omega_i \in V' \}$ defines the event in which some state in $V'$ is reached.
The outcome of the game, denoted by $\out_v(\strat{1}, \strat{2})$ is the set of possible paths
of $\Stochgame{G}^{\strat{1}, \strat{2}}_{v}$ starting at vertex $v$ (i.e., the possible plays when strategies $\strat{1}$ and $\strat{2}$ are used).  When the initial state $v$ is fixed,  we write
$\out(\strat{1}, \strat{2})$ instead of $\out_{v}(\strat{1}, \strat{2})$. 

A \emph{Boolean objective} for $\mathcal{G}$ is a  set $\Phi \subseteq \Omega$, we say that a play $\omega$ is \emph{winning} for Player $1$ at vertex $v$ if $\omega \in \Phi$, otherwise we say that it is winning for Player $2$ 
(i.e., we consider \emph{zero-sum} games).   A strategy $\strat{1}$ is a \emph{sure winning strategy} for Player $1$  from vertex $v$ if, for every strategy $\strat{2}$ for Player $2$, $\out_v(\strat{1}, \strat{2}) \subseteq \Phi$. $\strat{1}$ is said to be \emph{almost-sure winning} if for every strategy $\strat{2}$ for Player $2$, we have $\Prob{\strat{1}}{\strat{2}}_{\Stochgame{G},v}(\Phi)=1$.
	Sure  and almost-sure winning strategies for Player $2$ are defined in a similarly.
Reachability games are games with Boolean objectives of the style: 
$\Diamond V'$, for some set $V' \subseteq V$. A standard result is that, if a reachability game has a sure winning strategy, then
it has a pure memoryless sure winning strategy \cite{ChatterjeeH12}. 

A \emph{quantitative objective} is a measurable  function $f: \Omega \rightarrow \mathbb{R}$. Given a measurable function we define $\Expect{\strat{1}}{\strat{2}}_{\Stochgame{G},v}[f]$ as the expectation of
function $f$ under probability $\Prob{\strat{1}}{\strat{2}}_{\Stochgame{G},v}$. The goal of Player $1$ is to maximize the value $f$ of the play, whereas the goal of 
Player $2$ is to minimize it. Sometimes quantitative objective functions can be defined via \emph{rewards}. These are assigned by a \emph{reward function} $r:V \rightarrow \mathbb{R}$. A \emph{stochastic game with rewards} is a structure $(V, E , V_1, V_2, V_\Probabilistic, \delta, r)$ composed of a stochastic game and a reward function.
	The value of the game for Player $1$ under strategy $\strat{1}$ at vertex $v$, 
denoted $\val_1(\strat{1})(v)$, is defined as the infimum over all the values  
resulting from Player $2$ strategies when the game starts at $v$, i.e., $\val_1(\strat{1})(v) = \inf_{\strat{2} \in \Strategies{2}} \Expect{\strat{1}}{\strat{2}}_{\Stochgame{G},v}[f]$.
The \emph{value of the game} for Player $1$ from vertex $v$ is defined as the supremum of the values of all Player $1$ strategies, i.e., $\sup_{\strat{1} \in \Strategies{1}} \val_1(\strat{1})(v)$.
Analogously, the value of the game for a Player $2$ strategy $\strat{2}$ and the value of the game 
for Player $2$ are defined as $\val_2(\strat{2})(v) = \sup_{\strat{1} \in \Strategies{1}}  \Expect{\strat{1}}{\strat{2}}_{\Stochgame{G},v} [f]$ 
and $\inf_{\strat{2} \in \Strategies{2}} \val_2(\strat{2})(v)$, respectively. We say that a game is determined if both values are equal, that is,
$\sup_{\strat{1} \in \Strategies{1}} \val_1(\strat{1})(v) = \inf_{\strat{2} \in \Strategies{2}} \val_2(\strat{2})(v)$, for every vertex $v$.

\section{Probabilistic Masking Simulation} \label{sec:prob_mask_dist}

In this section we introduce probabilistic masking simulation which is
a probabilistic extension of the masking simulation relation
introduced in \cite{CastroDDP18b}.  We also give a symbolic version of
the stochastic game characterization of the relation, and provide  an
algorithm to solve it.

\subsection{The relation.}%
Roughly speaking, a probabilistic masking simulation is a relation
between PTSs that extends probabilistic bisimulation~\cite{Larsen91}
in order to account for fault masking.  Intuitively, one of the PTSs
acts as the nominal model (i.e., the specification), while the other
one models the implementation of the system under faults.  The
nominal model describes the ideal behavior of the system (i.e., when
no faults are considered), while the implementation describes a
fault-tolerant version of the system, where the occurrence of faults
are taken into account and a fault tolerance mechanism is expected to
act upon them.  Probabilistic masking simulation allows one to analyze
whether the implementation is able to mask the faults while
preserving the behavior of the nominal model.  More specifically, for
non-faulty transitions, the relation behaves as probabilistic
bisimulation, which is captured by means of couplings and relations
respecting these couplings (as done for instance in \cite{Larsen91}).
The novel part is given by the occurrence of faults: if the
implementation performs a fault, the nominal model matches it by doing
nothing.

For a set
of actions $\Sigma$, and a (finite) set of \emph{fault labels} $\faults$, with $\faults\cap\Sigma=\emptyset$, we define $\SigmaF
= \Sigma \cup \faults$.  Intuitively, the elements of $\faults$
indicate the occurrence of a fault in a faulty implementation.
Furthermore, when useful we consider the set $\Sigma^i = \{ e^i
\mid e \in \Sigma\}$, containing the elements of $\Sigma$ indexed with
superscript $i$.

\begin{definition} \label{def:masking_rel}
  Let $A =( S, \Sigma, {\rightarrow}, s_0 )$ and
  $A' =( S', \SigmaF, {\rightarrow'}, s_0' )$ be two PTSs representing
  the nominal and the implementation model, respectively.
  $A'$ is \emph{strong probabilistic masking fault-tolerant} with
  respect to $A$ if there exists a relation $\M \subseteq S \times S'$
  such that:
  \begin{inparaenum}[(a)]
  \item%
    $s_0 \M s'_0$, and
  \item%
    for all $s \in S, s' \in S'$ with $s \M s'$ and all $e \in \Sigma$
    and $F \in \faults$ the following holds:
  \end{inparaenum}%
  \begin{enumerate}[(1)]
  \item%
    if $s \xrightarrow{e} \mu$, then $s' \xrightarrow{e} \mu'$ and
    $\mu \MaskCoup \mu'$ for some $\mu'$;
  \item%
    if $s' \xrightarrow{e} \mu'$, then $s \xrightarrow{e} \mu$ and
    $\mu \MaskCoup \mu'$ for some $\mu$;
  \item%
    if $s' \xrightarrow{F} \mu'$, then $\Dirac_s \MaskCoup \mu'$.
  \end{enumerate}
  If such relation exists we say that $A'$ is a \emph{strong
    probabilistic masking fault-tolerant implementation} of $A$, 
    denoted $A \Masking A'$.
\end{definition}

\begin{example}\label{example:memory}
Consider a memory cell storing one bit of information that periodically refreshes its value.  The memory supports writing and reading operations, whereas a refresh performs a read operation and overwrites the value with itself.
Obviously, in this system, the result of a reading depends on the value stored in the cell. 
Thus, a property associated with the system is that the value read from the cell coincides with that of the last performed writing.
This is captured by the nominal model given at the left of
Figure~\ref{fig:exam_1_mem_cell} in  {\PRISM} 
notation~\cite{DBLP:conf/cav/KwiatkowskaNP11}.  Actions $\texttt{r}i$
and $\texttt{w}i$ (for $i=0,1$) represent the actions of reading or
writing  value $i$.  The bit stored in the memory is saved in
variable \texttt{b}.  A \texttt{tick} action indicates the passing of
one time unit and in doing so, with probability \texttt{p}, it
enables the refresh action (\texttt{rfsh}).  Variable \texttt{m}
indicates whether the system is in write/read mode, or producing a
refresh.

\begin{figure}[t]
\begin{minipage}[t]{.45\textwidth}
\fontsize{6.6}{6.6}\selectfont\ttfamily
\begin{tabbing}
x\=xxxxxxx\=xxxxxxxxxxxxx\=x\=xxx\= \kill    
module NOMINAL\\[1ex]
\>b : [0..1] init 0;\\
\>m : [0..1] init 0; \>\>// 0 = normal,\\
\>                   \>\>// 1 = refreshing\\[1ex]
\>[w0]   \>(m=0)          \>\>-> \>(b'= 0);\\
\>[w1]   \>(m=0)          \>\>-> \>(b'= 1);\\
\>[r0]   \>(m=0) \& (b=0) \>\>-> \>true;\\
\>[r1]   \>(m=0) \& (b=1) \>\>-> \>true;\\
\>[tick] \>(m=0)          \>\>-> \>p : (m'= 1) +\\
\>       \>               \>\>   \>(1-p) : true;\\
\>[rfsh] \>(m=1)          \>\>-> \>(m'= 0);\\[1ex]
endmodule\\
\end{tabbing}
\end{minipage}
\hfill
\begin{minipage}[t]{.55\textwidth}
\fontsize{6.6}{6.6}\selectfont\ttfamily
\begin{tabbing}
x\=xxxxxxxx\=xxxxxxxxxxxxx\=xxx\=xxx\= \kill    
module FAULTY\\[1ex]
\>v : [0..3] init 0;\\
\>s : [0..2] init 0; \>\>// 0 = normal, 1 = faulty,\\
\>                   \>\>// 2 = refreshing\\
\>\textcolor{red}{f : [0..1] init 0;} \>\>\textcolor{red}{// fault limiting artifact}\\[1ex]
\>[w0]    \>(s!=2)           \>\>-> \>(v'= 0) \& (s'= 0);\\
\>[w1]    \>(s!=2)           \>\>-> \>(v'= 3) \& (s'= 0);\\
\>[r0]    \>(s!=2) \& (v<=1) \>\>-> \>true;\\
\>[r1]    \>(s!=2) \& (v>=2) \>\>-> \>true;\\
\>[tick]  \>(s!=2)           \>\>-> \>p : (s'= 2) + q : (s'= 1) \\
\>        \>                 \>\>   \>+ (1-p-q) : true;\\
\>[rfsh]  \>(s=2)            \>\>-> \>(s'=0)\\
\>        \>                 \>\>   \>\& (v'= (v<=1) ? 0 : 3);\\
\>[fault] \>(s=1) \textcolor{red}{\& (f<1)}   \>\>-> \>(v'= (v<3) ? (v+1) : 2) \\
\>        \>                 \>\>   \>\& (s'= 0) \textcolor{red}{\& (f'= f+1)};\\
\>[fault] \>(s=1) \textcolor{red}{\& (f<1)}    \>\>-> \>(v'= (v>0) ? (v-1) : 1) \\
\>        \>                 \>\>   \>\& (s'= 0) \textcolor{red}{\& (f'= f+1)};\\[1ex]
endmodule\\
\end{tabbing}
\end{minipage}

\caption{Nominal and fault-tolerant models for the memory cell.} \label{fig:exam_1_mem_cell}
\end{figure}

A potential fault in this scenario occurs when a cell unexpectedly changes its value (e.g.,  as a consequence of some electromagnetic interference).  In practice, the occurrence of such an error 
has a certain probability. A typical technique to deal with this situation is \emph{redundancy}; for instance, using three memory bits instead of one. Then, writing operations are performed simultaneously 
on the three bits while reading returns the value read by majority (i.e.,  by \emph{voting}).
The right hand-side model of
Figure~\ref{fig:exam_1_mem_cell} represents this implementation with
the occurrence of the fault implicitly modeled (ignore, by the time
being the red part).  Variable \texttt{v} counts the votes for the value
1.  Thus writing 1 (\texttt{w1}) sets \texttt{v} in 3, and writing 0
(\texttt{w0}) sets it in 0.  The read actions would return 1
(\texttt{r1}) if $\texttt{v}\geq 2$ and 0 (\texttt{r0}) otherwise.
In addition to enabling the refresh action, a \texttt{tick} may also
enable the occurrence of a fault with probability \texttt{q}, with the
restriction that $\texttt{p}+\texttt{q}\leq 1$.
Variable \texttt{s} indicates that the system is in normal mode
($\texttt{s}=0$), in a state where a fault may occur ($\texttt{s}=1$),
or producing a refresh ($\texttt{s}=2$).  Notice that reading and
writing are allowed as long as the system is not producing a refresh.
The red coloured text of the figure is an artifact to limit the number
of faults to 1.  Under this condition, it is easy to check that the
relation
$\M = \{{\langle(b,m),(v,s,f)\rangle} \mid {{2b\leq v\leq 2b{+}1} \wedge (m=1 \Leftrightarrow s=2)}\}$
is a probabilistic masking simulation (where $b$, $m$, $v$, $s$, and
$f$ represent the values of variables \texttt{b}, \texttt{m},
\texttt{v}, \texttt{s}, and \texttt{f}, respectively.)

It should be evident that if the red coloured text is not present then
\texttt{FAULTY} is not a probabilistic masking fault-tolerant
implementation of \texttt{NOMINAL}. 
\end{example}
\subsection{A characterization in terms of a stochastic game.}

In the following, we define a stochastic masking simulation game for the nominal model 
$A = ( S, \Sigma, {\rightarrow}, s_0 )$ and the implementation model $A' = ( S', \Sigma_{\mathcal{F}}, {\rightarrow'}, s'_0 )$. 
The game is similar to a bisimulation game \cite{Stirling98}, and it is played by two players, 
named by convenience the Refuter ($\Refuter$) and the Verifier ($\Verifier$). The Verifier 
wants to prove that $s \in S$ and $s' \in S'$ are probabilistic masking similar, 
and the Refuter intends to disprove that.
The game starts from the pair of states $(s,s')$ and the following steps are repeated:
\begin{enumerate}
\item[1)]
  $\Refuter$ chooses either a transition $s \xrightarrow{a} \mu$ from
  the nominal model or a transition $s' \xrightarrow{a} \mu'$ from the
  implementation;
\item[2.a)]
  If $a \notin \faults$, $\Verifier$ chooses a transition matching
  action $a$ from the opposite model, i.e., a transition $s'
  \xrightarrow{a} \mu'$ if $\Refuter$'s choice was from the nominal,
  or a transition $s \xrightarrow{a} \mu$ otherwise.  In addition,
  $\Verifier$ chooses a coupling $w$ for $(\mu, \mu')$;
\item[2.b)]
  If $a \in \faults$, $\Verifier$ can only selects the Dirac
  distribution $\Dirac_{s}$ and the only possible coupling $w$ for
  $(\Dirac_{s}, \mu')$;
\item[3)]
  The successor pair of states $(t, t')$ is chosen probabilistically
  according to $w$.
\end{enumerate}
If the play continues forever, then the Verifier wins; otherwise, the
Refuter wins.  Step 2.b is the only one that differs from the usual
bisimulation game.  It is necessary for the asymmetry produced by
transitions that represent the occurrence of faults: if the Refuter
chooses to play a fault in the implementation, then it needs to be
masked, and therefore the Verifier cannot produce any move in the
nominal model.  Thus, the probabilistic step of the fault can only be
matched by a Dirac distribution on the same state of the nominal
model.

In the following we define the stochastic making game graph that
allows us to formalize this game.

\begin{definition} \label{def:strong_masking_game_graphi}
  Let $A =( S, \Sigma, {\rightarrow}, s_0 )$ and
  $A' = ( S', \Sigma_\faults , {\rightarrow'}, s'_0 )$ be two PTSs.
  The 2-players \emph{stochastic masking game graph}
  $\StochG_{A,A'} = (V^\StochG, E^\StochG, V^\StochG_\Refuter, V^\StochG_\Verifier, V^\StochG_\Probabilistic, \InitVertex, \delta^\StochG)$,
   is defined as follows:
  \begin{align*}
    V^\StochG = \
    & V^\StochG_\Refuter \cup V^\StochG_\Verifier \cup V^\StochG_\Probabilistic, \text{where: }\\
    V^\StochG_\Refuter = \
    & \{ (s, \mhyphen, s', \mhyphen, \mhyphen, \mhyphen, \Refuter) \mid
          s \in S \wedge s' \in S' \} \cup
      \{\ErrorSt\}\\
    V^\StochG_\Verifier = \
    & \{ (s, \sigma^1, s', \mu, \mhyphen, \mhyphen, \Verifier) \mid
         s \in S \wedge s' \in S'
         \wedge \sigma \in \Sigma
         \wedge (s, \sigma, \mu) \in {\rightarrow} \} \cup {} \\
    & \{ (s, \sigma^2, s', \mhyphen, \mu', \mhyphen, \Verifier) \mid
         s \in S \wedge s' \in S'
         \wedge \sigma \in \SigmaF
         \wedge (s', \sigma, \mu') \in {\rightarrow'}\}\\
    V^\StochG_\Probabilistic = \
    & \{ (s, \mhyphen, s', \mu, \mu', w, \Probabilistic) \mid
         s \in S \wedge s' \in S' \wedge 
         \mu \in \Dist(S) \wedge \mu' \in \Dist(S')
         \wedge w \in \couplings{\mu}{\mu'}\}\\
    \InitVertex = \
    & ( s_0, \mhyphen, s'_0, \mhyphen, \mhyphen, \mhyphen, \Refuter )
    \text{ \ (the Refuter starts playing)}\\
    & \hspace{-2.6em}
    \delta^\StochG : V^\StochG_\Probabilistic \rightarrow \Dist(V^\StochG_\Refuter),
      \text{ defined by } 
      \delta^\StochG((s, \mhyphen, s', \mu, \mu', w, \Probabilistic))((t, \mhyphen, t', \mhyphen, \mhyphen, \mhyphen, \Refuter)) = w(t,t')
      \text{,}
  \end{align*}
  and $E^\StochG$ is the minimal set satisfying the following rules:
  \begin{align*}
    s \xrightarrow{\sigma} \mu
    & \Rightarrow \tuple{(s, \mhyphen, s', \mhyphen, \mhyphen, \mhyphen, \Refuter), (s, \sigma^{1}, s', \mu, \mhyphen, \mhyphen, \Verifier)}\in E^\StochG
    \tag{1$_1$}\label{play:r:1}\\
    s' \xrightarrowprime{\sigma} \mu'
    & \Rightarrow \tuple{(s, \mhyphen, s', \mhyphen, \mhyphen, \mhyphen, \Refuter),(s, \sigma^{2}, s', \mhyphen, \mu', \mhyphen, \Verifier)}\in E^\StochG
    \tag{1$_2$}\label{play:r:2}\\
    {s' \xrightarrowprime{\sigma} \mu'} \wedge {w \in \couplings{\mu}{\mu'}}
    & \Rightarrow \tuple{(s, \sigma^1, s', \mu, \mhyphen, \mhyphen, \Verifier),(s, \mhyphen, s', \mu, \mu', w, \Probabilistic)}\in E^\StochG
    \tag{2.a$_1$}\label{play:v:1}\\
    {\sigma \notin \faults} \wedge {s \xrightarrow{\sigma} \mu} \wedge {w \in \couplings{\mu}{\mu'}}
    & \Rightarrow \tuple{(s, \sigma^2, s', \mhyphen, \mu', \mhyphen, \Verifier), (s, \mhyphen, s', \mu, \mu', w, \Probabilistic)}\in E^\StochG
    \tag{2.a$_2$}\label{play:v:2a}\\
    {F \in \faults} \wedge {w \in \couplings{\Dirac_s}{\mu'}}
    & \Rightarrow \tuple{(s, F^2, s', \mhyphen, \mu', \mhyphen, \Verifier), (s, \mhyphen, s', \Dirac_s, \mu', w, \Probabilistic)}\in E^\StochG
    \tag{2.b}\label{play:v:2b}\\
    (s, \mhyphen, s', \mu, \mu', w, \Probabilistic) \in V^\StochG_\Probabilistic \wedge & \\ 
    (t,t') \in \support{w} 
    & \Rightarrow \tuple{(s, \mhyphen, s', \mu, \mu', w, \Probabilistic), (t, \mhyphen, t', \mhyphen, \mhyphen, \mhyphen, \Refuter)}\in E^\StochG
    \tag{3}\label{play:p}\\
    {v\in (V^\StochG_\Verifier{\cup}\{\ErrorSt\})} \wedge ({\nexists v'} & {}\neq\ErrorSt : {\tuple{v,v'}\in E^\StochG})
    \Rightarrow  \tuple{v,\ErrorSt}\in E^\StochG
    \tag{err}\label{play:err}
  \end{align*}
\end{definition}

The definition follows the idea of the game previously described.  A
round of the game starts in the Refuter's state $\InitVertex$.  Notice
that, at this point, only the current states of the nominal and
implementation models are relevant (all other information is not yet defined in this round and hence marked with
``$\mhyphen$'').  Step 1 of the game is encoded in rules
(\ref{play:r:1}) and (\ref{play:r:2}), where the Refuter chooses a
transition, thus defining the action and distribution that need to be
matched, and moving to a Verifier's state.  Thus, a Verifier's state
in $V^\StochG_\Verifier$ is a tuple that has also defined which action
and distribution need to be matched and which model the Refuter has
moved.  Step 2.a of the game is given by rules (\ref{play:v:1}) and
(\ref{play:v:2a}) in which the Verifier chooses a matching move from
the opposite model (hence defining the other distribution) and an
appropriate coupling, and moving to a probabilistic state.  Step 2.b
of the game is encoded in rule (\ref{play:v:2b}).  Here the Verifier
has no choice since it is obliged to choose the Dirac distribution
$\Dirac_s$ and the only available coupling in
$\couplings{\Dirac_s}{\mu'}$.  A probabilistic state in
$V^\StochG_\Probabilistic$ has everything defined to probabilistically
resolve the next step through function $\delta^\StochG$ (rule
(\ref{play:p})).
Finally, if a player has no move, then it can only move to the error
state $\ErrorSt$ (rule (\ref{play:err})). This can only happen in a Verifier's state or in
$\ErrorSt$.

The notion of probabilistic masking simulation can be captured by the
corresponding stochastic masking game with the appropriate Boolean
objective.

\begin{theorem} \label{thm:wingame_strat}
  Let $A= ( S, \Sigma, {\rightarrow}, s_0 )$ and $A'=( S', \SigmaF, {\rightarrow'}, s_0' )$ be two PTSs.
  We have $A \Masking A'$ iff the Verifier has a sure winning strategy for the stochastic masking game graph 
  $\mathcal{G}_{A,A'}$ with the Boolean objective $\neg \Diamond \ErrorSt$.
\end{theorem}

Notice that the graph for a stochastic masking game could be infinite.
Indeed, each probabilistic node of the graph includes a coupling
between the two contending distributions, and there can be uncountably
many of them.
It nonetheless induces an algorithm as
follows.  We define regions $W^i$ of the graph vertices.  Intuitively,
each $W^i$ represents a collection of vertices from which the Refuter
has a strategy (in the infinite game) with probability greater than
$0$ of reaching the error state in at most $i$ steps (these sets can
be thought of as a probabilistic version of \emph{attractors}
\cite{Jurd11}).

\begin{definition}\label{def:W}
  Let
  $\StochG_{A,A'} = (V^\StochG, E^\StochG, V^\StochG_\Refuter, V^\StochG_\Verifier, V^\StochG_\Probabilistic, \InitVertex, \delta^\StochG)$
  be a stochastic masking game graph for PTSs $A$ and $A'$.
  We define sets $W^i$ (for $i \geq 0$)  as follows:
  \begin{align*}
    W^0= {} & \{\ErrorSt\}, \\
    W^{i+1} = {}
    & \textstyle
      \{v' \mid v' \in V^\StochG_\Refuter \wedge \post(v') \cap W^i \neq \emptyset \} \cup {} \\
    & \textstyle
      \{v' \mid {v' \in V^\StochG_\Verifier} \wedge  {\vertices{\post(v')} \subseteq  \bigcup_{j \leq i}W^{j}} \wedge {\vertices{\post(v')} \cap W^i \neq \emptyset} \} \cup {} \\
    & \textstyle
      \{v' \mid v' \in V^\StochG_\Probabilistic \wedge \sum_{v'' \in \post(v') \cap W^i} \delta^\StochG(v')(v'') > 0 \}
  \end{align*}
  where
  $\vertices{V'}= \{ (s, \mhyphen, s', \mu, \mu', w, \Probabilistic) \in V' \cap V^\StochG_\Probabilistic \mid w \in \vertices{\couplings{\mu}{\mu'}}\}$.
  Finally, let $W = \bigcup_{i  \geq 0} W^i$.
\end{definition}
The sets $W^i$ can be used to solve the game $\StochG_{A,A'}$.
Notice, in particular, that we do not take into account all possible
couplings but only those that are  vertices of the polytope
$\couplings{\pr{3}{v}}{\pr{4}{v}}$. ($\pr{i}{(x_0,\dots,x_n)}$ is the $(i+1)$-th projection, i.e., \ $x_i$.)
%
This is sufficient to determine the winner of the game since every coupling in $\couplings{\pr{3}{v}}{\pr{4}{v}}$ can be expressed as a convex combination of its vertices.  Thus, if there is a positive probability of reaching the error state with some coupling, there is also a positive probability of reaching it through a vertex coupling.
By taking only the vertex couplings, only a finite number of graph
vertices are collected in each $W_i$ and hence $W$ can be effectively
computed with a fix point algorithm.

The following result is a straightforward adaptation of the results
for reachability games over finite graphs \cite{ChatterjeeH12}.
	
\begin{theorem}\label{th:strat-W} 
  Let
  $\StochG_{A,A'} = (V^\StochG, E^\StochG, V^\StochG_\Refuter, V^\StochG_\Verifier, V^\StochG_\Probabilistic, \InitVertex, \delta^\StochG)$
  be a stochastic masking game graph for PTSs $A$ and $A'$.  Then, the
  Verifier has a sure winning strategy from vertex $v$ iff
  $v \notin W$.
\end{theorem}

It is worth noting that, for a probabilistic vertex
$(s, \mhyphen, s', \mu, \mu', w, \Probabilistic) \in V^\StochG_\Probabilistic$,
the two-way transportation polytope $\couplings{\mu}{\mu'}$ has at
least $\frac{\max\{m,n\}!}{(\max\{m, n\}-\min\{m,n\}+1)!}$ vertices
(and at most $m^{n-1}n^{m-1}$ vertices)~\cite{KleeWitzgall}, where
$m=\support{\mu}$ and $n=\support{\mu'}$.  Therefore, it could be
computationally impractical to calculate such sets.

\subsection{A Symbolic Game Graph.}


In this section, we introduce a finite representation of stochastic
masking games through a symbolic representation which enables a more
efficient algorithm.
We define the symbolic graph for a stochastic masking game in two
parts.  The first part captures the non-stochastic behaviour of the
game by removing the stochastic choice ($\delta^\StochG$) of the
game graph, as well as the couplings on the vertices.  The
second part appends an equation system to each probabilistic vertex
whose solution space is the polytope defined by the set of all
couplings for the contending distributions.

\begin{definition} \label{def:symbolic_game_graph}
  Let $A = ( S, \Sigma, {\rightarrow}, s_0 )$
  and $A' = ( S', \Sigma_\faults, {\rightarrow'}, s'_0 )$
  be two PTSs.
  The \emph{symbolic game graph} for the stochastic masking game
  $\mathcal{G}_{A,A'}$ is defined by the structure
  $\SymbG_{A,A'} = ( V^{\SymbG}, E^{\SymbG}, V^{\SymbG}_\Refuter, V^{\SymbG}_\Verifier, V^{\SymbG}_\Probabilistic, v_0^{\SymbG} )$,
  where:
  \begin{align*}
    V^\SymbG = \
    & V^\SymbG_\Refuter \cup V^\SymbG_\Verifier \cup V^\SymbG_\Probabilistic, \text{ where: }\\
    V^\SymbG_\Refuter = \
    & \{ (s, \mhyphen, s', \mhyphen, \mhyphen, \Refuter) \mid
          s \in S \wedge s' \in S' \} \cup
      \{\ErrorSt\}\\
    V^\SymbG_\Verifier = \
    & \{ (s, \sigma, s', \mu, \mhyphen, \Verifier) \mid
         s \in S \wedge s' \in S'
         \wedge \sigma \in \Sigma^1
         \wedge (s, \sigma, \mu) \in {\rightarrow} \} \cup {} \\
    & \{ (s, \sigma, s', \mhyphen, \mu', \Verifier) \mid
         s \in S \wedge s' \in S'
         \wedge \sigma \in \SigmaF^2
          \wedge (s', \sigma, \mu') \in {\rightarrow} \} \\
    V^\SymbG_\Probabilistic = \
    & \{ (s, \mhyphen, s', \mu, \mu', \Probabilistic) \mid
         s \in S \wedge s' \in S' \wedge 
         \mu \in \Dist(S) \wedge \mu' \in \Dist(S')\}\\
    v_0^{\SymbG} = \
    & ( s_0, \mhyphen, s'_0, \mhyphen, \mhyphen, \Refuter ),
  \end{align*}
  and $E^\SymbG$ is the minimal set satisfying the following rules:
  \begin{align*}
    s \xrightarrow{\sigma} \mu
    & \Rightarrow \tuple{(s, \mhyphen, s', \mhyphen, \mhyphen, \Refuter), (s, \sigma^{1}, s', \mu, \mhyphen, \Verifier)}\in E^\SymbG \\
    s' \xrightarrowprime{\sigma} \mu'
    & \Rightarrow \tuple{(s, \mhyphen, s', \mhyphen, \mhyphen, \Refuter),(s, \sigma^{2}, s', \mhyphen, \mu', \Verifier)}\in E^\SymbG \\
    {s' \xrightarrowprime{\sigma} \mu'}
    & \Rightarrow \tuple{(s, \sigma^1, s', \mu, \mhyphen, \Verifier),(s, \mhyphen, s', \mu, \mu', \Probabilistic)}\in E^\SymbG \\
    {\sigma \notin \faults} \wedge {s \xrightarrow{\sigma} \mu}
    & \Rightarrow \tuple{(s, \sigma^2, s', \mhyphen, \mu', \Verifier), (s, \mhyphen, s', \mu, \mu', \Probabilistic)}\in E^\SymbG\\
    {F \in \faults}
    & \Rightarrow \tuple{(s, F^2, s', \mhyphen, \mu', \Verifier), (s, \mhyphen, s', \Dirac_s, \mu', \Probabilistic)}\in E^\SymbG \\
     (s, \mhyphen, s', \mu, \mu', \Probabilistic) \in V^\SymbG_\Probabilistic \wedge & \\ 
   ((t,t') \in \support{\mu}\times \support{\mu'})
    & \Rightarrow \tuple{(s, \mhyphen, s', \mu, \mu', \Probabilistic), (t, \mhyphen, t', \mhyphen, \mhyphen,  \Refuter)}\in E^\SymbG \\
    {v\in (V^\SymbG_\Verifier{\cup}\{\ErrorSt\})} & {} \wedge {(\nexists {v'\neq\ErrorSt} : {\tuple{v,v'}\in E^\SymbG})}
    \Rightarrow  \tuple{v,\ErrorSt}\in E^\SymbG
  \end{align*}
  In addition, for each
  $v=(s, \mhyphen, s', \mu, \mu', \Probabilistic) \in V^\SymbG_\Probabilistic$,
  define the set of variables
  $X(v)=\{x_{s_i,s_j} \mid s_i \in \text{Supp}(\mu) \wedge s_j \in \text{Supp}(\mu')\}$,
  and the system of equations
  \begin{align*}
    \Eq(v) = {}
    & \textstyle
    \big\{ \sum_{s_j \in \support{\mu'}} x_{s_k,s_j}=\mu(s_k) \mid s_k \in \support{\mu} \big\} \cup {} \\
    & \textstyle
    \big\{ \sum_{s_k \in \support{\mu}} x_{s_k,s_j}=\mu'(s_j) \mid s_j \in \support{\mu'} \big\} \cup {} \\
    & \textstyle
    \big\{ x_{s_k,s_j} \geq 0 \mid s_k \in \support{\mu} \wedge s_j \in \support{\mu'} \big\}
  \end{align*}
\end{definition}

Notice that $\{\bar{x}_{s_k,s_j}\}_{s_k,s_j}$ is a solution of
$\Eq(v)$ if and only if there is a coupling
$w\in\couplings{\mu}{\mu'}$ such that $w(s_k,s_j)=\bar{x}_{s_k,s_j}$
for all $s_k \in \support{\mu}$ and $s_j \in \support{\mu'}$.

In addition, given a set of game vertices
$V' \subseteq V^\SymbG_\Refuter$,
we define $\Eq(v, V')$ by extending $\Eq(v)$ with an equation limiting  the couplings in such a way that vertices in $V'$ are \emph{not} reached.
Formally,
$\Eq(v, V') = \Eq(v) \cup \big\{\sum_{(s, \mhyphen, s', \mhyphen, \mhyphen, \Refuter) \in V'} x_{s,s'} = 0\big\}$.
%
By properly defining a family of sets $V'$, we will show that
probabilistic masking simulation can be checked in polynomial time
through the symbolic game graph.

In the following we propose to use the symbolic game graph to solve
the infinite game. By doing so, we obtain a polynomial time procedure.
Similarly to Definition~\ref{def:W}, we provide an inductive
construction of the Refuter-winning nodes using equation systems in
place of sets of polytope vertices, as follows.

\begin{definition}\label{def:U}
  Let
  $\SymbG_{A,A'} = ( V^{\SymbG},  E^{\SymbG}, V^{\SymbG}_\Refuter, V^{\SymbG}_\Verifier, V^{\SymbG}_\Probabilistic, v_0^{\SymbG} )$
  be a symbolic game graph for PTSs $A$ and $A'$.
  Sets $U^i$ (for $i \geq 0$) are defined as follows:
  \begin{align*}
    U^0 = {} & \{\ErrorSt\},  \label{def:Uji}\\
    U^{i+1} = {}
    & \textstyle 
    \{v' \mid v' \in V^\SymbG_\Refuter \wedge \post(v') \cap U^{i} \neq \emptyset \} \cup {}\\
    & \textstyle 
    \{v' \mid v' \in V^\SymbG_\Verifier \wedge \post(v') \subseteq \bigcup_{j\leq i}U^{j} \wedge \post(v') \cap U^i \neq \emptyset  \}  \cup {}\\
    & \textstyle 
    \{v' \mid v' \in V^\SymbG_\Probabilistic \wedge \post(v') \cap U^{i} \neq \emptyset \wedge \Eq(v', \post(v') \cap U^{i}) \text{ has no solution}\}
  \end{align*}
  Furthermore, we define $U = \bigcup_{i \geq 0} U^i$.
\end{definition}

The construction of each $U^{i+1}$ follows a similar idea
as the construction of $W^{i+1}$, only varying significantly in the
case of the probabilistic vertices.  The first and second line
correspond to the Refuter and Verifier players, respectively.  The last
one corresponds to the probabilistic player.
Notice that, if $\Eq(v', \post(v') \cap U^{i})$ has no solution, it
means that every possible coupling will inevitably lead with some
probability to a ``losing'' state of a smaller level, since, in
particular, equation
$\sum_{(s, \mhyphen, s', \mhyphen, \mhyphen, \Refuter) \in (\post(v') {\cap} U^{i})} x_{s,s'} = 0$
cannot be satisfied.

There is a strong connection between sets $W^i$ and $U^i$: a vertex is
in $W^i$ if and only if its abstract version is in $U^i$.  This is
formally stated in the next theorem.
	
\begin{theorem}\label{th:U-and-W}
  Let $\mathcal{G}_{A,A'}$ be a stochastic masking game graph
  for PTSs $A$ and $A'$ and let $\SymbG_{A,A'}$
  be the corresponding symbolic game graph.
  For every $v \in V^\StochG$, $u \in V^\SymbG$ such that
  $\pr{i}{v} = \pr{i}{u}$ (for $0 \leq i \leq 4$) and
  $\pr{6}{v} = \pr{5}{u}$, and for all $k\geq 0$,
  $v \in U^k$ if and only if $u \in W^k$.
\end{theorem}

The following theorem is a direct consequence of Theorems
\ref{th:strat-W} and \ref{th:U-and-W}.

\begin{theorem}
  Let $\StochG_{A,A'}$ be a stochastic game graph for PTSs $A$
  and $A'$, and let $\SymbG_{A,A'}$ be the corresponding symbolic game
  graph.  Then, the Verifier has a sure winning strategy in
  $\StochG_{A,A'}$ if and only if $v^\SymbG_0 \notin U$.
\end{theorem}

As a consequence of this last theorem and
Theorem~\ref{thm:wingame_strat}, it suffices to calculate the set $U$
over $\SymbG_{A,A'}$ to decide whether there is a probabilistic
masking simulation between $A$ and $A'$.
This can be done in polynomial time, since $\Eq(v,C)$ can be solved in
polynomial time and the number of iterations to construct $U$ is
bounded by $|V^\SymbG|$.  Since $V^\SymbG$ linearly depends on the
transitions of the involved PTSs, we have the following theorem.

\begin{theorem}
  Let $A$ and $A'$ be PTSs.  $A \Masking A'$ can be decided in time
  $O(\textit{Poly}(m\cdot m'))$, where $m$ and $m'$ are 
  the size of the transitions of $A$ and $A'$, respectively.
\end{theorem}

\section{Quantifying Fault Tolerance in Almost-Sure Failing Systems} \label{sec:prob_almost_sure}

Probabilistic masking simulation determines whether a fault-tolerant
implementation is able to completely mask faults.  However, in
practice, this kind of masking fault-tolerance is uncommon. Usually,
fault-tolerant systems are able to mask a number of faults before
suffering a failure. Our main goal in this section is to extend the
game theory presented in the previous section to be able to measure
the amount of masking tolerance exhibited by a system before it fails.
To do this, we extend the probabilistic masking games with a
quantitative objective function.  This is done in such a way that the
expected value of this function indicates the number of ``milestones''
that the fault-tolerant implementation is expected to cross before
failing.  A milestone is any interesting event that may occur during
the execution of the system.  For instance, a milestone may be the
number of faults occurring during the execution of the system, and
therefore this measure will reflect the number of faults that are
tolerated by the system before crashing. Another milestone may be
successful acknowledgments in a transmission protocol, in this case we
measure the expected number of chunks that the protocol is able to
transfer under the occurrence of faults before failing. 

To do this type of measuring,  we need  stochastic games with 
quantitative objective functions.  Intuitively, these objective functions count the
number of ``milestones'' observed during a play.
Therefore, we first extend stochastic masking games as follows.

\begin{definition}%
  Let
  $A =( S, \Sigma, \rightarrow, s_0 )$ and $A' =( S', \SigmaF, \rightarrow', s_0' )$
  be two PTSs.
  A \emph{stochastic masking game graph with milestones} $\mathcal{M}$
  is a tuple
  $\MilestoneG_{A,A'} = (V^\StochG, E^\StochG, V^\StochG_\Refuter, V^\StochG_\Verifier, V^\StochG_\Probabilistic, \InitVertex, \delta^\StochG,  \mathcal{M})$
  where:
  \begin{inparaenum}[(i)]
  \item%
    $(V^\StochG, E^\StochG, V^\StochG_\Refuter, V^\StochG_\Verifier, V^\StochG_\Probabilistic, \InitVertex, \delta^\StochG)$  is a stochastic masking game graph,
  \item%
    $\mathcal{M} \subseteq \SigmaF^2$ is the set of \emph{milestones}, and
  \item%
    $\reward^\StochG(v) = \chi_{\mathcal{M}}(\pr{1}{v})$ is a \emph{reward function}.
  \end{inparaenum}
\end{definition}
In this definition, $\chi_B$ is the characteristic function over
set $B$ defined as usual:
$\chi_B(a) = 1$ whenever $(a{\in}B)$ and $\chi_B(a) = 0$ otherwise.
If $B=\{b\}$ is a singleton set, we simply write $\chi_b$.

Given a stochastic masking game graph with milestones and reward function $r^\StochG$, for any play 
$\rho = \rho_0, \rho_1,  \dots$,
we define the \emph{masking payoff function} by
$\FMask(\rho) = \lim_{n \rightarrow \infty} (\sum^{n}_{i=0} \reward^\StochG(\rho_i))$.

	Intuitively, the payoff function $\FMask$ characterizes the number of milestones that a fault-tolerant implementation is able to achieve until the
error state is reached. This type of payoff functions are usually called \emph{total rewards} in the literature.
    One may think of this as a game played by the fault-tolerance built-in mechanism and a (malicious) player that chooses the 
way in which faults occur. In this game, the Verifier is the maximizer (she intends to obtain as much milestones as possible) and the 
Refuter is the minimizer (she intends to prevent the Verifier from achieving milestones).


For the game of expected total reward to be determined, we need that
the stochastic game is almost-sure stopping, i.e., that the game
reaches a sink vertex with probability 1~\cite{FilarV96}. We manage to
extend the determinacy property to games that are almost-sure stopping
under the condition that the minimizing player is
fair~\cite{DBLP:journals/corr/abs-2112-09811}.

In our setting, this amounts to considering \emph{almost-sure failing
masking games}, that is, games in which the error state $\ErrorSt$ is
reached with probability 1.  Moreover, we require that the Refuter
plays fair.  This is necessary to avoid the Refuter stalls the game
in an unproductive loop.
Indeed, consider the scenario described in
Example~\ref{example:memory} and set the stochastic masking game
between the nominal and faulty model of
Figure~\ref{fig:exam_1_mem_cell} (omitting the red part).
One would expect that the game leads to a failure with probability 1.
However, the Refuter has strategies for which the probability of
reaching $\ErrorSt$ is less than $1$.  For instance, the Refuter may
always play the reading action, and hence the Verifier has to mimic
this action forever, this yields a probability of $0$ of reaching the
error state.
Observe that, in this scenario, the Refuter is behaving in a
benevolent manner, playing with the aim to avoid the error state.
Clearly, this is against the spirit of the behaviour of faults which
one expects they happen if waiting long enough.
Therefore the assumption that the Refuter is fair in the sense that,
if some action or fault is infinitely often enabled for the Refuter,
it will eventually play such action or fault.

The setting for the stochastic game with the masking payoff function
as objective that we present in the rest of the section stands
on~\cite{DBLP:journals/corr/abs-2112-09811}, only that here special
care is needed due to the infinite nature of the stochastic game
graphs.
For this reason we limit the results of the rest of the section to
(randomized) memoryless strategies and postpone the general result for
further work.
Thus, we let $\Strategies{\Verifier}^{\Memoryless}$ and
$\Strategies{\Refuter}^{\Memoryless}$ denote the sets of all
(randomized) memoryless strategies for the Verifier and the Refuter,
respectively, and similarly, we let
$\Strategies{\Verifier}^{\Memoryless\Deterministic}$ and
$\Strategies{\Refuter}^{\Memoryless\Deterministic}$ denote the sets of
all pure (or deterministic) memoryless strategies.

A Refuter's fair play is defined as a play in which the Refuter commits to follow a strong fair pattern, i.e., that includes infinitely often any transition that is enabled infinitely often.  A fair strategy for the Refuter, is a strategy that always measures 1 on the set of all the Refuter's fair plays, regardless of the strategy of the Verifier.  The definition provided below follows the style in~\cite{DBLP:journals/dc/BaierK98,BaierK08,DBLP:journals/corr/abs-2112-09811}.

\begin{definition}
  Given a masking game
  $\StochG_{A,A'} = (V^\StochG, E^\StochG, V^\StochG_\Refuter, V^\StochG_\Verifier, V^\StochG_\Probabilistic, \InitVertex, \delta^\StochG)$,
  the \emph{set of all Refuter's fair plays} is defined by
  $ 
	\RFP = \{ \omega \in \Omega \mid v \in \inf(\omega) \cap V^\StochG_\Refuter \Rightarrow \post(v) \subseteq \inf(\omega) \}
  $.
  A Refuter strategy $\strat{\Refuter}$ is said to
  be \emph{almost-sure fair} iff, for every Verifier's strategy
  $\strat{\Verifier}$,
  $\Prob{\strat{\Refuter}}{\strat{\Verifier}}_{\StochG_{A,A'},\InitVertex}(\RFP) = 1$.
\end{definition}


Under this concept, the stochastic masking game is almost-sure failing
under fairness if for every Verifier's strategy and every Refuter's
fair strategy, the game leads to an error with probability 1.  This is
formally defined as follows.

\begin{definition}
  Let $A$ and $A'$ be two PTSs.  We say that the stochastich masking
  game $\StochG_{A,A'}$ is \emph{almost-sure failing under fairness
  (and memoryless strategies)} iff, for every memoryless strategy
  $\strat{\Verifier} \in \Strategies{\Verifier}^{\Memoryless}$ and any
  fair memoryless strategy
  $\strat{\Refuter} \in \Strategies{\Refuter}^{\Memoryless}$,
  $\Prob{\strat{\Verifier}}{\strat{\Refuter}}_{\StochG_{A,A'}, \InitVertex}(\Diamond \ErrorSt)=1$.
\end{definition}



    Interestingly,  under strong fairness assumptions the determinacy of games is preserved \cite{DBLP:journals/corr/abs-2112-09811}. Furthermore,  in finite stochastic games with fairness restrictions the value of the game can be computed by calculating the greatest fixed point of the following Bellman functional. We can adapt this result to our games by using the vertices of the polytopes when computing the   values of the probabilistic states. 
\begin{theorem}\label{theo:det-fairness} Let $\MilestoneG_{A,A'}$ be a stochastic game with milestones for some PTSs $A$ and $A'$ that is almost-sure 
failing for fair Refuter's strategies.  Then,  we have:
\[
 \inf_{\strat{\Refuter} \in \Strategies{\Refuter}^{\Memoryless}} \sup_{\strat{\Verifier} \in \Strategies{\Verifier}^{\Memoryless}}  \mathbb{E}^{\strat{\Verifier},\strat{\Refuter}}_{\MilestoneG_{A,A'}, v}[\FMask] 
    = \sup_{\strat{\Verifier \in \Strategies{\Verifier}^{\Memoryless}}} \inf_{\strat{\Refuter \in \Strategies{\Refuter}^{\Memoryless}}}    \mathbb{E}^{\strat{\Verifier},\strat{\Refuter}}_{\MilestoneG_{A,A'}, v}[\FMask] 
    < \infty.
\]
Moreover, 
the value of the game for memoryless strategies for the Verifier and  fair memoryless Refuter's strategies is the greatest fixpoint of the following functional $\Bellman$: 
\[
    \Bellman(f)(v) =
    \begin{cases}
           \displaystyle \min \{\upperbound,  \max_{w \in \vertices{\couplings{\pr{3}{v}}{\pr{4}{v}}}} \{ \chi_{\mathcal{M}}(\pr{1}{v})  + \sum_{v' \in \post(v)} w(\pr{0}{v'},\pr{2}{v'})  f(v') \} \}& \text{ if }v \in V^{\SymbG}_\Probabilistic  \\
           \displaystyle \min \{ \upperbound, \max \{\chi_{\mathcal{M}}(\pr{1}{v})  + f(v') \mid v' \in \post(v) \} \} & \text{ if } v \in  V^{\SymbG}_\Verifier, \\
           \displaystyle \min \{ \upperbound,  \min \{\chi_{\mathcal{M}}(\pr{1}{v})  + f(v') \mid v' \in \post(v) \} \} & \text{ if } v \in  V^{\SymbG}_\Refuter, \\
           \displaystyle 0 & \text{ if } v=\ErrorSt.
           \displaystyle 
    \end{cases}
\]
where $\upperbound$ is a number such that 
$\upperbound \geq \inf_{\strat{\Refuter} \in \Strategies{\Refuter}^{\Memoryless\Deterministic}} \sup_{\strat{\Verifier} \in \Strategies{\Verifier}^{\Memoryless\Deterministic}} \Expect{\strat{\Verifier}}{\strat{\Refuter}}_{\MilestoneG_{A,A'}, \InitVertex}[\FMask]$.
\end{theorem}

   Also, we can check whether  a game is almost-sure failing under fairness by computing predecessor sets in the symbolic game graph.  To do so, we define the symbolic version of the predecessor sets.
    Given a game $\StochG_{A,A'}$ and its symbolic version $\SymbG_{A,A'}$, let $\SymbEFairpre(C)$ and 
$\SymbAFairpre(C)$, for a given set $C$ of symbolic vertices, be defined as follows:
\begin{align*}
	\SymbEFairpre(C) ={}& \{ v \in V^\SymbG_\Probabilistic \mid  \exists v' \in C\cap V^\SymbG_\Refuter : \pr{0}{v'}\in\support{\pr{3}{v}} \land \pr{2}{v'}\in\support{\pr{4}{v}} \} \\
		       & \cup \{ v \in  V^\SymbG_\Verifier \cup V^\SymbG_\Refuter  \mid \exists v' \in C : (v,v') \in E^\SymbG_{A,A'} \}\\
	\SymbAFairpre(C) = {}&\{ v \in V^\SymbG_\Probabilistic \mid \Eq(v,C) \textit{  has no solution }\} \\
		      & \cup \{ v \in  V^\SymbG_\Verifier   \mid \forall v' \in C : (v,v') \in E^\SymbG \Rightarrow v' \in C \}\\
		      &  \cup \{ v \in  V^\SymbG_\Refuter  \mid \exists v' \in C : (v,v') \in E^\SymbG \}
\end{align*}
%
In particular, the first set in the definition of $\SymbEFairpre(C)$
collects all probabilistic vertices $v$ for which there is a coupling
that leads to a Refuter vertex $v'$ in $C$.  For this is sufficient to
check that the states $\pr{0}{v'}$ and $\pr{2}{v'}$ that define $v'$
are on the respective support sets of the probabilities
$\pr{3}{v}$ and $\pr{4}{v}$ that define $v$ (since
it is always possible to define a coupling that assigns positive
probability to a pair of states in the respective support sets).  The
first set in the definition of $\SymbAFairpre(C)$ collects all
probabilistic vertices $v$ for which there is no coupling ``avoiding'' $C$, that is, no coupling that leads with probability 0 to the set of all
pair of states defining a vertex in $C$.  A coupling avoiding
$C$ will solve $\Eq(v,C)$.
By using $\SymbEFairpre$ and $\SymbAFairpre$ recursively, we can
decide whether a game is almost-sure failing under fairness as follows.

\begin{theorem}\label{theo:decide-stopping} Given a masking game $\StochG_{A,A'}$ and its symbolic version $\SymbG_{A,A'}$,  we have that 
 $\StochG_{A,A'}$ is almost-sure failing under fairness  iff
$
  \InitVertex \in V^\SymbG \setminus {\SymbEFairpre}^*(V^\SymbG \setminus {\SymbAFairpre}^*(\ErrorSt)),
$
where $\InitVertex$ is the initial state of $\SymbG_{A,A'}$ and $V^\SymbG$ its sets of vertices.
\end{theorem}

As $\Eq(v,C)$ can be computed in polynomial time, so do the predecessor
sets $\SymbEFairpre(C)$ and $\SymbAFairpre(C)$.  As a consequence, the
problem of deciding whether a stochastic masking game is almost-sure
failing under fairness is also polynomial.

\section{Experimental Evaluation} \label{sec:experimental_eval}
We implemented the approach described in this paper in a prototype tool, available at \cite{PMaskD}.
Tables~\ref{table:resultsMEM}, \ref{table:resultsNMR} and \ref{table:resultsNMR2} report the results obtained  for three case studies:
 a Redundant Cell Memory (our running example); N-Modular Redundancy (NMR), a standard example of fault-tolerance \cite{ShoomanBook}; and a NMR processor/memory architecture with N voters \cite{KrishnaBook}.
In the tables, $M_{t}$ and $M_{r}$ are used to indicate the measurement results for the tick and refresh actions, considered as milestones, respectively.

Some words are useful to interpret the results. For the memory cell example,  either increasing the redundancy, or augmenting the frequency of refreshing, have positive effects in the measures.  In practice, these values can be taken into account when designing a fault-tolerant component that provides an optimal balance between fault-tolerance and hardware costs. For example, assuming a fault probability of $0.05$,  one might prefer $3$ bits and more frequent refreshing, over $5$ bits with a less often refreshing, despite the software overhead.
NMR consists of N modules that independently perform a task, and whose results are processed by a perfect voter to produce a single output.
These modules may exhibit an unexpected behavior with a given probability, in which case they output an incorrect value. The results for this case study are similar to the memory cell example when there is $0$ probability of refreshing. 
The last case study consists of N processors that output a value to a memory module through N voters.  Both the voters and processors may output an incorrect value with  certain probability.
The experiment outputs the same results if the fault probability of voters and processors are exchanged. This suggests that, provided that the probability that a pair processor-voter fails remains the same, the system is more tolerant when faults occur equally distributed on voters and processors.

We have run our experiments on a MacBook Air with processor 1.3 GHz Intel Core i5 and 
a memory of 4 Gb. The tool and case studies are available in the tool repository~\cite{PMaskD}.
\begin{table}[t]
  \centering
  \parbox{.45\linewidth}{
  \scalebox{0.7}{
    \begin{tabular}{c c c c c}
      Bits & Fault Prob. & Refresh Prob. & $M_{t}$ & $M_{r}$  \\ \midrule
      \multirow{9}{*}{3} 
                  & \multirow{3}{*}{$0.5$} 
                  & $0.5$ & $6$ & $3$\\ 
                  & & $0.1$ & $4.44$ & $0.44$\\ 
                  & & $0.05$ & $4.2$  & $0.21$\\ \cline {2-5}
                  & \multirow{3}{*}{$0.1$} 
                  & $0.5$ & $70$ & $35$\\ 
                  & & $0.1$ & $30$ & $3$\\ 
                  & & $0.05$ & $25$  & $1.25$\\ \cline {2-5}
                  & \multirow{3}{*}{$0.05$} 
                  & $0.5$ & $240.02$ & $120.01$\\ 
                  & & $0.1$ & $80$ & $8$\\ 
                  & & $0.05$ & $60$  & $3$\\ \cline {1-5}
      \multirow{9}{*}{5} 
                  & \multirow{3}{*}{$0.5$} 
                  & $0.5$ & $14$ & $7$\\ 
                  & & $0.1$ & $7.28$ & $0.72$\\ 
                  & & $0.05$ & $6.62$ & $0.33$\\ \cline {2-5}
                  & \multirow{3}{*}{$0.1$} 
                  & $0.5$ & $430$ & $215.02$\\ 
                  & & $0.1$ & $70$ & $7$\\ 
                  & & $0.05$ & $47.5$ & $2.38$\\ \cline {2-5}
                  & \multirow{3}{*}{$0.05$} 
                  & $0.5$ & $2660.1$ & $1330.06$\\ 
                  & & $0.1$ & $260.01$ & $26.01$\\ 
                  & & $0.05$ & $140$ & $7$\\ \cline {1-5}
      \multirow{9}{*}{7} 
                  & \multirow{3}{*}{$0.5$} 
                  & $0.5$ & $30$ & $15$\\ 
                  & & $0.1$ & $10.74$ & $1.07$\\ 
                  & & $0.05$ & $9.28$ & $0.46$\\ \cline {2-5}
                  & \multirow{3}{*}{$0.1$} 
                  & $0.5$ & $2590.1$ & $1295.05$\\ 
                  & & $0.1$ & $150$ & $15$\\ 
                  & & $0.05$ & $81.25$ & $4.06$\\ \cline {2-5}
                  & \multirow{3}{*}{$0.05$} 
                  & $0.5$ & $29281$ & $14640.5$\\ 
                  & & $0.1$ & $800.03$ & $80.03$\\ 
                  & & $0.05$ & $300$ & $15$\\ \bottomrule
    \end{tabular}
    }
    \vspace{0.2cm}
    \caption{Experimental results on a Redundant Memory Cell.}
    \label{table:resultsMEM}
    }
    \hspace{0.5cm}
    \centering
    \parbox{.45\linewidth}{\centering
    \scalebox{0.7}{
    \begin{tabular}{c c c}
      Modules & Fault Prob. & $M_{t}$ \\
      \midrule
      \multirow{3}{*}{$3$}  & $0.5$ & $4$ \\ 
      & $0.1$ & $20$ \\ 
      & $0.05$ & $40$ \\ \cline{1-3}
      \multirow{3}{*}{$5$} & $0.5$ & $6$ \\ 
      & $0.1$ & $30$ \\ 
      & $0.05$ & $60$ \\ \cline{1-3}
      \multirow{3}{*}{$7$} & $0.5$ & $8$ \\
      & $0.1$ & $40$ \\ 
      & $0.05$ & $80$ \\ \bottomrule
      \hline
    \end{tabular}
    }
    \vspace{0.2cm}
    \caption{Experimental results on a N-Modular Redundant System.}
    \label{table:resultsNMR}
    \scalebox{0.7}{
    \begin{tabular}{c c c c}
      N & P.Fault Prob. & V.Fault Prob. & $M_{t}$ \\ 
      \midrule
      \multirow{3}{*}{$3$}  & $0.09$  & $0.01$ & $21.8$ \\ 
      & $0.07$ & $0.03$ & $24.2$ \\ 
      & $0.05$ & $0.05$ & $25$ \\ \cline{1-4}
      \multirow{3}{*}{$5$} & $0.09$ & $0.01$ & $33.18$ \\ 
      & $0.07$ & $0.03$ & $38.94$ \\ 
      & $0.05$ & $0.05$ & $41.25$ \\ \cline{1-4}
      \multirow{3}{*}{$7$} & $0.09$ & $0.01$ & $44.39$ \\
      & $0.07$ & $0.03$ & $53.78$ \\ 
      & $0.05$ & $0.05$ & $58.12$ \\ \cline{1-4}
      \multirow{3}{*}{$9$} & $0.09$ & $0.01$ & $55.53$ \\
      & $0.07$ & $0.03$ & $68.58$ \\ 
      & $0.05$ & $0.05$ & $75.39$ \\ \bottomrule
      \hline
    \end{tabular}
    }
    \vspace{0.2cm}
    \caption{Experimental results on a NMR processor/memory architecture with N voters.}
    \label{table:resultsNMR2}
    }
    \vspace{-.5cm}
\end{table}

\section{Related Work} \label{sec:related_work}
	The games introduced in \cite{Bacci0LM17,BacciBLMTB19,DesharnaisGJP04,DesharnaisLT11}  are based on probabilistic bisimulation, so they are symmetric. Furthermore, in \cite{DesharnaisGJP04,DesharnaisLT11} the nodes of the game graph are modeled using subsets of states of the PTSs, in our formulation
we do not use subsets of states. The games defined in \cite{Bacci0LM17,BacciBLMTB19} use Kanterovich's and Hausdorff's liftings to deal with probabilistic distributions and non-determinism, respectively. In addition, the authors use the vertices of the transportation polytopes to model the probabilistic vertices. In contrast, we introduced a symbolic representation of games to avoid the state explosion caused by the vertices of the polytopes. Also note that the metrics introduced in \cite{Bacci0LM17,BacciBLMTB19} measure the (probabilistic) bisimulation distance between two PTSs, which for almost-sure failing systems is always $1$.
		
Another related framework is  defined in \cite{LanotteMT17}. Therein, the authors introduce a notion of weak simulation quasimetric tailored for reasoning about the evolution of \emph{gossip protocols}. This makes it possible to compare network protocols that have similar behaviour up to a certain tolerance; being $0$ and $1$ the minimum and maximum distance, respectively.  Note that using this quasimetric to compare a network protocol with an almost-sure failing implementation will always return $1$, thus that approach cannot be used to quantify the masking fault-tolerance of almost-sure failing systems.

After the case studies of Section~\ref{sec:experimental_eval},
\emph{Mean-Time To Failure} (MTTF)~\cite{ReliabilityBook} may come to
mind.  Though this metric (lifted to games) may be the result of a
particular case study, we present a much more general framework.
Indeed, on the one hand, we do not necessarily have to count time
units, and other events may be set as milestones.  On the other hand,
the computation of MTTF would normally require the identification of
failures states in an ad hoc manner, while we do this at a higher
level of abstraction: the failure situation appears in the game as a
result of comparing the implementation model against the nominal
model.


\section{Conclusions and Future Work} \label{sec:conclusions}

We presented a relation of masking fault-tolerance between
probabilistic transition systems, which is accompanied by a
corresponding probabilistic game characterization. Even though the
game could be infinite, we proposed an alternative finite symbolic
representation by means of which the game can be solved in polynomial
time.
We extended the game with quantitative objectives based on counting
``milestones'' thus providing a way to quantify the amount of masking
fault tolerance provided by a given implementation.
As this game inherits the characteristic of total reward objectives,
some stopping criterion is necessary and thus the game is required to
be almost-sure failing under a fair Refuter.
By restricting to (randomize) memoryless strategies, we could show
that the resulting game is determined and can be computed by solving a
collection of functional equations.  We also provided a polynomial
technique to decide whether a game is almost-sure failing under
fairness.

There are many directions for future work. As an immediate one, we
have pending to extend the result on quantitative objectives to
non-memoryless strategies.  Given the result
in~\cite{DBLP:journals/corr/abs-2112-09811}, we believe this is
possible but special care is needed to deal with the infinite nature
of the game.
In a different direction, in this paper we introduced a strong version
of probabilistic masking simulation.  However, for analyzing
non-trivial systems, a weak version of this kind of relation is needed
since it could abstract away from internal transitions which are
mostly associated with the fault-tolerant machinery of the
implementation.
%
Besides, we have so far only worked with masking fault-tolerance.
Similar ideas to those presented in this paper could be extrapolated
to other levels of fault-tolerance like fail-safe and non-masking.
Finally, we have presented a prototype tool for measuring some
well-known small case studies.  Our goal is to develop an automated tool
to support the use of these measurements in practice.

\newpage
\bibliography{references}

\newpage
\appendix \label{sec:appendix}

\section{Proofs of Properties}

\noindent
\textbf{Proof of Theorem \ref{thm:wingame_strat}.}
  Let $A= ( S, \Sigma, \rightarrow, s_0 )$ and $A'=( S', \SigmaF, \rightarrow', s_0' )$ be two PTSs.
  We have $A \Masking A'$ iff the Verifier has a sure winning strategy for the stochastic masking game graph 
  $\mathcal{G}_{A,A'}$ with the Boolean objective 
$\Phi = \{ \omega_0,\omega_1, \dots  \in \Omega \mid \forall {i \geq 0} : \omega_i \neq \ErrorSt \}$. \\
\noindent
\begin{proof} 
\noindent ``Only If'': Assume $A \Masking A'$, thus there is a probabilistic masking simulation  $\M \subseteq S \times S'$.
	Let us define a sure winning strategy $\strat{\Verifier}$ for the Verifier as follows.
Given a state $(s, \sigma^1, s', \mu, \mhyphen, \mhyphen, \Verifier)$ (resp. $(s, \sigma^2, s', \mhyphen, \mu', \mhyphen, \Verifier)$), if $s \M  s'$, $\strat{\Verifier}$ selects a transition 
$((s, \sigma^1, s', \mu, \mhyphen, \mhyphen, \Verifier), (s, \sigma^1, s', \mu, \mu', w, \Probabilistic))$ (resp. $((s, \sigma^2, s', \mhyphen, \mu', \mhyphen, \Verifier),(s, \sigma^2, s', \mu, \mu', w, \Probabilistic))$) such that $w$ is a $\M$-respecting coupling for ($\mu,\mu'$) 
(which is guaranteed to exist by Def.~\ref{def:masking_rel}). Otherwise, $\strat{\Verifier}$ selects an arbitrary vertex. 
Let us show that this strategy is sure winning for the Verifier in the initial state. 
We have to prove that, for any Refuter's strategy $\strat{\Refuter}$, we have $\out(\strat{\Verifier}, \strat{\Refuter}) \subseteq\Omega \setminus \Phi$, where $\out(\strat{\Verifier}, \strat{\Refuter})$ denotes the set of paths generated when strategies $\strat{\Refuter}$ and $\strat{\Verifier}$ are used.  Let $\strat{\Refuter}$ be any strategy for the Refuter, and $\omega = \omega_0, \omega_1,  \dots$ the corresponding play in $\text{out}(\strat{\Verifier}, \strat{\Refuter})$. 
We prove by induction that $\forall i \geq 0:  \omega_i \neq \ErrorSt \wedge (\pr{6}{\omega_i}\neq \Refuter \Rightarrow \pr{0}{\omega_i} \M  \pr{2}{\omega_i})$. For $i=0$, the proof is straightforward. Assume that the property holds for $\omega_i$, if $\omega_i$ is a Verifier's vertex and 
$\pr{1}{\omega} = \sigma^1$ (resp. $\sigma^2$) with $\sigma \notin \faults$, then by definition of $\strat{\Verifier}$ and Def.~\ref{def:masking_rel} 
$\omega_{i+1} =  (s, \sigma^1, s', \mu, \mu', w,  \Probabilistic)$
(resp. $\omega_{i+1} =  (s, \sigma^2, s', \mu, \mu', w, \Probabilistic)$). Thus, we have by inductive hypothesis that $s \M s'$ and also that 
$\omega_{i+1} \neq \ErrorSt$. If $\pr{1}{\omega_i} \in \faults$,
then the proof is similar, but taking into account that  $\mu = \Dirac_s$.
If $\omega_i$ is a Refuter's vertex, then $\omega_{i+1}$ is a Verifier's vertex, and it cannot be $\ErrorSt$ because by construction only Verifier's nodes are adjacent to the $\ErrorSt$.
If $\omega_i$ is a probabilistic vertex, then note that $\support{\pr{5}{\omega_i}} \neq \emptyset$ and therefore $\ErrorSt \neq \omega_{i+1}$. 
Thereby, $ \pr{0}{\omega_{i+1}}  \M \pr{2}{\omega_{i+1}}$. 
That is, $\forall i \geq 0 : \omega_i \neq \ErrorSt$. Hence, $\omega \in \Omega \setminus \Phi$.

``If'': Suppose that the Verifier has a sure winning strategy $\strat{\Verifier}$
from the initial state. Then, we define a probabilistic masking simulation relation 
as follows: $\M =\{(s,s') \mid (s, \mhyphen, s', \mu, \mu', w, \Probabilistic) \in \strat{\Verifier}(V^{\StochG}_\Verifier)$ for some sure winning strategy $\strat{\Verifier} \}$. 
We know by our assumption that this set is not empty and it is direct to
see that $(s_0,s'_0) \in \M$. 
First, let us prove that for any $(s, \mhyphen, s', \mu, \mu', w, \Probabilistic) \in \strat{\Verifier}(V^{\StochG}_\Verifier)$ we have  $\mu \MaskCoup \mu'$. 
Assume $(s, \mhyphen, s', \mu, \mu', w, \Probabilistic) \in \strat{\Verifier}(V^{\StochG}_\Verifier)$ and it is not the case that $\mu \MaskCoup \mu'$, or equivalently 
$\exists t,t' : w(t,t') > 0 \wedge \neg (t \M t')$. 
Thus, we have a successor $(t, \mhyphen,t', \mhyphen,\mhyphen,\mhyphen, \Refuter)$ of
$(s, \mhyphen, s', \mu, \mu', w, \Probabilistic)$ which can be chosen with probability greater than $0$ and  $(t,t') \notin \M$. 
Furthermore, there exists a $t \xrightarrow{\sigma} t_0$ 
(or $t' \xrightarrowprime{\sigma} t'_0$) such that $(t_0, \sigma^1, t', \mu, \mhyphen, \mhyphen, \Verifier) \in \post((t, \mhyphen,t', \mhyphen,\mhyphen,\mhyphen, \Refuter))$ 
(resp. $(t, \sigma^2, t'_0, \mhyphen,  \mu', \mhyphen, \Verifier) \in \post((t, \mhyphen,t', \mhyphen,\mhyphen,\mhyphen, \Refuter))$). 
This state cannot be $\ErrorSt$ and there must be a winning strategy for the Verifier from it; otherwise, the Refuter would have a winning strategy from $(t, \mhyphen,t', \mhyphen,\mhyphen,\mhyphen, \Refuter)$, and
$(s, \mhyphen, s', \mu, \mu', w, \Probabilistic) \notin \image{\pi}$ for some winning strategy $\strat{\Verifier}$ for the Verifier. 
But then $(t,t') \notin \M$, contradicting the assumption above. Thus, $\mu \MaskCoup \mu'$.

Let us now prove that $\M$  is a probabilistic masking simulation. Assume that $s \M s'$ and $s \xrightarrow{a} \mu$, for any successor 
$(s, a, s', \mu, \mhyphen,  \mhyphen, \Verifier)$ of $(s, \mhyphen, s', \mhyphen, \mhyphen,  \mhyphen, \Refuter)$ we have a  sure winning strategy $\strat{\Verifier}$ 
for $\Verifier$, such that: $\pi((s, a, s', \mu, \mhyphen,  \mhyphen, \Verifier)) = (s, a, s', \mu, \mu',  w, \Probabilistic)$. 
That is, there is a $s' \xrightarrow{a} \mu'$ and also by the property proven above, we have $\mu \MaskCoup \mu'$. Similarly for the cases: $s' \xrightarrow{a} \mu'$ and $s' \xrightarrow{F} \mu'$ for $F \in \faults$. Finally, since 
$\Verifier$ has a winning strategy from $(s_0, \mhyphen, s'_0, \mhyphen, \mhyphen,  \mhyphen, \Refuter)$ we have that $s_0 \M s'_0$. 
Thus, all the requirements of Def.~\ref{def:masking_rel} holds, and so $\M$ is a probabilistic masking relation. 
\end{proof} \\

\noindent
\textbf{Proof of Theorem \ref{th:strat-W}.}
Let $\StochG_{A,A'} = (V^\StochG, E^\StochG, V^\StochG_\Refuter, V^\StochG_\Verifier, V^\StochG_\Probabilistic, \InitVertex, \delta^\StochG)$ 
be a stochastic masking game graph for some PTSs $A$ and $A'$, we have that the Verifier has a sure winning strategy from vertex 
$v$ iff $v \notin W$. \\
\noindent
\begin{proof} First, we can define a $2$-player reachability game obtained from  $\StochG_{A,A'}$ by considering the probabilistic nodes as Refuter's nodes, and ignoring the probabilistic distribution, let $\mathcal{H}^\StochG_{A,A'}$ be that game. 
It is clear that a Verifier's strategy is sure winning in $\StochG_{A,A'}$ iff this strategy is winning in $\mathcal{H}^\StochG_{A,A'}$. 
Then, proving the theorem reduces to show that the sets $W^i$ determines the winning strategies of the Verifier in 
$\mathcal{H}^\StochG_{A, A'}$ (recall that only the vertices of the polytopes are taken into account for defining the $W$'s). 
If the Verifier has a winning strategy from vertex $v$ let us prove that $v \notin W^k$ for every $k$ by induction. 
For $k=1$ it is direct. Now, assume that the property holds for $W^k$, let $v$ be an arbitrary vertex such that the Verifier 
has a winning strategy named $\strat{\Verifier}$ from $v$. 
If $v$ is a Verifier's node and $v \in W^{k+1}$, then $\vertices{\post(v))}  \subseteq W^k$. 
Thus, by inductive hypothesis $\strat{\Verifier}(v) \notin \vertices{\post(v)}$, that is, $\strat{\Verifier}(v)$ is a probabilistic vertex whose
coupling is not a vertex. Furthermore, it is a Refuter's node in $\mathcal{H}^\StochG_{A,A'}$. For this node we have $v' \in \post(\strat{\Verifier}(v))$ iff $\pr{5}{\strat{\Verifier}(v)}(\pr{0}{v'}, \pr{2}{v'})>0$. 
Note that $\pr{5}{\strat{\Verifier}(v)}$ is a point of the polytope defined by $\couplings{\pr{3}{\strat{\Verifier}(v)}}{\pr{4}{\strat{\Verifier}(v)}}$, 
since polytopes do not contain lines, either $\pr{5}{\strat{\Verifier}(v)}$ is a vertex
or there is a polytope's vertex $w'$ such that  $w'(\pr{3}{\strat{\Verifier}(v)}, \pr{4}{\strat{\Verifier}(v)})>0$ iff $\pr{5}{\strat{\Verifier}(v)}(\pr{3}{\strat{\Verifier}(v)}, \pr{4}{\strat{\Verifier}(v)})>0$. 
Thereby, there is a $v'' \in \vertices{\post(v)}$ such that $\post(v'') = \post(\strat{\Verifier}(v))$, that is, $v'' \in W^k$ implies 
that $\strat{\Verifier}(v) \in W^{k}$ which is a contradiction and so $v \notin W^{k+1}$. 

Let us define a strategy $\strat{\Verifier}$ which is winning strategy in $\mathcal{H}^\StochG_{A,A'}$  for any Verifier's node $v \notin W$. 
If $v \notin W$, then $\strat{\Verifier}(v) = v'$ for some $v' \in \post(v) \cap (V^\StochG \setminus W)$ (which is guaranteed to exist by assumption), 
moreover, if $v \in W$, then $\strat{\Verifier}(v) = v'$ for an arbitrary node $v'$. 
Let us prove that for any play generated by $\strat{\Verifier}$: $v_0, v_1, \dots$ we have $v_i \notin W$, the proof is by induction on $i$.
For $i=0$ it is direct, assuming that $v_i \notin W$ let us prove that $v_{i+1} \notin W$. 
If $v_i$ is a Refuter's node by Def.~\ref{def:W}  
$\post(v_i) \cap W = \emptyset$, and therefore, $v_{i+1} \notin W$. 
If $v_i$ is Verifier's node, by definition of $\strat{\Verifier}$: $v_{i+1} = \strat{\Verifier}(v_i) \notin W$ and therefore the result follows.  
\end{proof} \\

\noindent
\textbf{Proof of Theorem \ref{th:U-and-W}.}
Given a stochastic masking game graph $\mathcal{G}_{A,A'}$ for some PTSs $A$ and $A'$ and the corresponding symbolic game $\SymbG_{A,A'}$. 
For any states $v \in V^\StochG$, $u \in V^\SymbG$ such that $\pr{i}{v} = \pr{i}{u}$ (for $0 \leq i \leq 4$) and for any $k>0$ 
we have that: $v \in U^k$ iff $u \in W^k$. \\
\noindent
\begin{proof}  
The proof is by induction on $k$. For k=1, we have that $W^1 = \{ \ErrorSt \} = U^1$. For the inductive case, consider arbitrary nodes $v \in V^\SymbG$ and $u \in V^{G}$,
such that $\pr{i}{v} = \pr{i}{u}$ for $0 \leq i \leq 4$. Note that these nodes also coincide in their last components, that is, 
either both are Refuter's nodes, Verifier's nodes, or probabilistic nodes. 
Assume these are Refuter's nodes, if $v \in U^{k}$ then $\post(v) \cap U^{k-1} \neq \emptyset$. 
Thus, there is some $v' \in \post(v) \cap U^{k-1}$ which is a probabilistic node. By Def.~\ref{def:U}, we have
many $u' \in \post(u)$ such that $\pr{i}{v'}=\pr{i}{u'}$ (for $0\leq i \leq 4$), and by induction we have that $u' \in W^{k-1}$ 
and therefore $u \in W^k$. 
Similarly, if $u \in W^k$ we have that $\post(u) \cap W^{k-1} \neq \emptyset$, and we proceed as before. 
If $v$ and $u$ are Verifier's nodes the proof is similar. 
Now, assume that  $v$ and $u$ are probabilistic nodes. 
If $v \in U^k$, then $\post(v) \cap U^{k-1} \neq \emptyset$. As proved above, we also have that $\post(u) \cap W^{k-1} \neq \emptyset$. 
Moreover, if $\Eq(v, \post(v) \cap U^{k-1})$ has no solutions, then we have at least a coupling (say $w$) for distributions 
$\pr{3}{v}$ and $\pr{4}{v}$ that satisfies the equations and also $w(\pr{0}{v'},\pr{2}{v'})>0$ for some $v' \in \post(v) \cap U^{k-1}$. 
By Def.~\ref{def:strong_masking_game_graphi}, we have a vertex $u' \in \post(u)$ such that $\pr{5}{u'}=w$, and thus $\sum_{u' \in \post(u) \cap W^{j-1}} \delta^\StochG(u)(u') >0$, which means that that $u' \in W^k$. 
Similarly, if $\sum_{u' \in \post(u) \cap W^{j-1}} \delta^\StochG(u)(u') >0$  then $\Eq(v, \post(v) \cap U^{k-1})$ has no solutions, and then
$u \in W^k$ implies $v \in U^k$.
\end{proof} \\

\noindent
\textbf{Proof of Theorem \ref{theo:det-fairness}.}  Let $\MilestoneG_{A,A'}$ be a stochastic game with milestones for some PTSs $A$ and $A'$ that is almost-sure 
failing for fair Refuter's strategies.  Then:
\[
    \inf_{\strat{\Refuter} \in \Strategies{\Refuter}^{\Memoryless}} \sup_{\strat{\Verifier} \in \Strategies{\Verifier}^{\Memoryless}}  \mathbb{E}^{\strat{\Verifier},\strat{\Refuter}}_{\MilestoneG_{A,A'}, v}[\FMask] 
    = \sup_{\strat{\Verifier \in \Strategies{\Verifier}^{\Memoryless}}} \inf_{\strat{\Refuter \in \Strategies{\Refuter}^{\Memoryless}}}    \mathbb{E}^{\strat{\Verifier},\strat{\Refuter}}_{\MilestoneG_{A,A'}, v}[\FMask] 
    < \infty
\]
Furthermore, the value of the game for memoryless strategies for the Verifier and  fair Memoryless Refuter's strategies is the greatest fixpoint of the following functional $\Bellman$: 
\[
    \Bellman(f)(v) =
    \begin{cases}
           \displaystyle \min \{\upperbound,  \max_{w \in \vertices{\couplings{\pr{3}{v}}{\pr{4}{v}}}} \{\sum_{v' \in \post(v)} w(\pr{0}{v'},\pr{2}{v'})  f(v') \} \}& \text{ if }v \in V^{\SymbG}_\Probabilistic  \\
           \displaystyle \min \{ \upperbound, \max \{\chi_{\mathcal{M}}(\pr{1}{v})  + f(v') \mid v' \in \post(v) \} \} & \text{ if } v \in  V^{\SymbG}_\Verifier, \\
           \displaystyle \min \{ \upperbound,  \min \{\chi_{\mathcal{M}}(\pr{1}{v})  + f(v') \mid v' \in \post(v) \} \} & \text{ if } v \in  V^{\SymbG}_\Refuter, \\
           \displaystyle 0 & \text{ if } v=\ErrorSt.
           \displaystyle 
    \end{cases}
\]
where $\upperbound$ is a number such that 
$\upperbound \geq \inf_{\strat{\Refuter} \in \Strategies{\Refuter}^{\Memoryless\Deterministic}} \sup_{\strat{\Verifier} \in \Strategies{\Verifier}^{\Memoryless\Deterministic}} \Expect{\strat{\Verifier}}{\strat{\Refuter}}_{\MilestoneG_{A,A'}, v}[\FMask]$, for every $v$.

\noindent
\begin{proof}
    First we prove that we can safely restrict to deterministic strategies when computing the value of the game for memoryless strategies.  To do so, 
we prove that for every memoryless strategies $\strat{\Verifier}$ and $\strat{\Refuter}$, there is a memoryless and deterministic strategy
$\strat{\Verifier}'$ such that: $\mathbb{E}^{\strat{\Verifier}',\strat{\Refuter}}_{\MilestoneG_{A,A'}, v}[\FMask]  \geq \mathbb{E}^{\strat{\Verifier},\strat{\Refuter}}_{\MilestoneG_{A,A'}, v}[\FMask]$. To do so, note that any memoryless strategy satisfies the following equation for every $v \in V^\StochG_\Verifier$:
\begin{align}
    \mathbb{E}_{\MilestoneG_{A,A'},v}^{\strat{\Verifier},\strat{\Refuter}}[\FMask]&  \leq \reward(v) + \sum_{v' \in \post(v)} \delta^{\strat{\Refuter},\strat{\Verifier}}(v,v')  \mathbb{E}_{\MilestoneG_{A,A'},v'}^{\strat{\Verifier},\strat{\Refuter}}[\FMask] \label{theo:det-fairness:eq3:l1}\\
    & \leq \reward(v) + \max_{v' \in \post(v)} \{  \mathbb{E}_{\MilestoneG_{A,A'},v'}^{\strat{\Verifier},\strat{\Refuter}}[\FMask] \}
    \label{theo:det-fairness:eq3:l12}
\end{align}
The first inequality  follows from the definition of expected value, the second inequality follows since $ \sum_{v' \in \post(v)} \delta^{\strat{\Refuter},\strat{\Verifier}}(v,v')  \mathbb{E}_{\MilestoneG_{A,A'},v'}^{\strat{\Verifier},\strat{\Refuter}}[\FMask] $ is a convex combination.  That is,  defining $\strat{\Verifier}'(v) = \argmax_{v' \in \post(v)}  \{ \mathbb{E}_{\MilestoneG_{A,A'},v'}^{\strat{\Verifier},\strat{\Refuter}}[\FMask] \}$, for every $v$,  we obtain
$\mathbb{E}_{\MilestoneG_{A,A'},v}^{\strat{\Verifier},\strat{\Refuter}}[\FMask] \leq \mathbb{E}_{\MilestoneG_{A,A'},v}^{\strat{\Verifier}',\strat{\Refuter}}[\FMask]$. 
    Similarly we can prove that for every memoryless strategies $\strat{\Refuter}$ and $\strat{\Verifier}$, there is a memoryless, deterministic and fair strategy 
$\strat{\Refuter}'$ such that $\mathbb{E}_{\MilestoneG_{A,A'},v}^{\strat{\Verifier},\strat{\Refuter}'}[\FMask] \leq \mathbb{E}_{\MilestoneG_{A,A'},v}^{\strat{\Verifier},\strat{\Refuter}}[\FMask]$.  These properties imply that:
\[
    \inf_{\strat{\Refuter} \in \Strategies{\Refuter}^{\Memoryless\Deterministic}}  \sup_{\strat{\Verifier} \in \Strategies{\Verifier}^{\Memoryless\Deterministic}} \mathbb{E}_{\MilestoneG_{A,A'},v'}^{\strat{\Verifier},\strat{\Refuter}}[\FMask] 
    =  \inf_{\strat{\Refuter} \in \Strategies{\Refuter}^{\Memoryless}}  \sup_{\strat{\Verifier} \in \Strategies{\Verifier}^{\Memoryless}} \mathbb{E}_{\MilestoneG_{A,A'},v'}^{\strat{\Verifier},\strat{\Refuter}}[\FMask]
\]
and similarly:
\[
    \sup_{\strat{\Verifier} \in \Strategies{\Verifier}^{\Memoryless\Deterministic}}  \inf_{\strat{\Refuter} \in \Strategies{\Refuter}^{\Memoryless\Deterministic}} \mathbb{E}_{\MilestoneG_{A,A'},v'}^{\strat{\Verifier},\strat{\Refuter}}[\FMask] 
    =  \sup_{\strat{\Verifier} \in \Strategies{\Verifier}^{\Memoryless}}  \inf_{\strat{\Refuter} \in \Strategies{\Refuter}^{\Memoryless}} \mathbb{E}_{\MilestoneG_{A,A'},v'}^{\strat{\Verifier},\strat{\Refuter}}[\FMask].
\]

    Now, we prove the theorem.  We define a restricted (finite) game just taking into account the vertices of the polytope defined by the couplings.  Consider the subgame $\mathcal{H}_{A,A'}$ obtained from $\MilestoneG_{A,A'}$ by restricting the successors of Verifier's vertices to the following
sets:
\begin{itemize}
	\item $\{ ((s, \sigma^2, s', \mhyphen, \mu', \mhyphen, \Verifier), (s, \mhyphen, s', \mu, \mu', w, \Probabilistic)) \mid (\exists\;\sigma \in \Sigma: s \xrightarrow{\sigma} \mu) \wedge   w \in \vertices{\couplings{\mu}{\mu'}}\} \subseteq E^\StochG$ for any $\sigma \notin \faults$,

 	 \item $\{ ((s, \sigma^1, s', \mu, \mhyphen, \mhyphen, \Verifier),(s, \mhyphen, s', \mu, \mu', w, \Probabilistic)) \mid (\exists\;\sigma \in \Sigma: s' \xrightarrowprime{\sigma} \mu' ) \wedge  w \in \vertices{\couplings{\mu}{\mu'}} \} \subseteq E^\StochG$,
	 
	 \item $\{ ((s, F^2, s', \mhyphen, \mu', \mhyphen, \Verifier), (s, \mhyphen, s', \Dirac_s, \mu', w, \Probabilistic)) \wedge w \in \vertices{\couplings{\Dirac_s}{\mu'}} \} \subseteq E^\StochG$ for any $F \in \faults$,
\end{itemize}
That is, we restrict the couplings to the vertices of the polytope $\couplings{\mu}{\mu'}$.  Note that since the set of vertices is finite, the game  $\mathcal{H}_{A,A'}$ is finite.  
We show now that:
\begin{equation}\label{eq:theo:determined:eq2}
    \sup_{\strat{\Verifier} \in \Strategies{\Verifier}^{\Memoryless}} \inf_{\strat{\Refuter} \in \Strategies{\Refuter}^{\Memoryless}} \mathbb{E}_{\mathcal{H}_{A,A'},v'}^{\strat{\Verifier},\strat{\Refuter}}[\FMask]
    \leq 
     \sup_{\strat{\Verifier} \in \Strategies{\Verifier}^{\Memoryless}} \inf_{\strat{\Refuter} \in \Strategies{\Refuter}^{\Memoryless}} \mathbb{E}_{\mathcal{G}_{A,A'},v'}^{\strat{\Verifier},\strat{\Refuter}}[\FMask],
\end{equation}
and:
\begin{equation}\label{eq:theo:determined:eq3}
    \inf_{\strat{\Refuter} \in \Strategies{\Refuter}^{\Memoryless}} \sup_{\strat{\Verifier} \in \Strategies{\Verifier}^{\Memoryless}} \mathbb{E}_{\mathcal{G}_{A,A'},v'}^{\strat{\Verifier},\strat{\Refuter}}[\FMask]
    \leq 
     \inf_{\strat{\Refuter} \in \Strategies{\Refuter}^{\Memoryless}} \sup_{\strat{\Verifier} \in \Strategies{\Verifier}^{\Memoryless}} \mathbb{E}_{\mathcal{H}_{A,A'},v'}^{\strat{\Verifier},\strat{\Refuter}}[\FMask],
\end{equation}
    Note that, by the property proven above, these are equivalent to:
\begin{equation}\label{eq:theo:determined:eq4}
    \sup_{\strat{\Verifier} \in \Strategies{\Verifier}^{\Memoryless\Deterministic}} \inf_{\strat{\Refuter} \in \Strategies{\Refuter}^{\Memoryless\Deterministic}} \mathbb{E}_{\mathcal{H}_{A,A'},v'}^{\strat{\Verifier},\strat{\Refuter}}[\FMask]
    \leq 
     \sup_{\strat{\Verifier} \in \Strategies{\Verifier}^{\Memoryless\Deterministic}} \inf_{\strat{\Refuter} \in \Strategies{\Refuter}^{\Memoryless\Deterministic}} \mathbb{E}_{\mathcal{G}_{A,A'},v'}^{\strat{\Verifier},\strat{\Refuter}}[\FMask],
\end{equation}
and:
\begin{equation}\label{eq:theo:determined:eq5}
    \inf_{\strat{\Refuter} \in \Strategies{\Refuter}^{\Memoryless\Deterministic}} \sup_{\strat{\Verifier} \in \Strategies{\Verifier}^{\Memoryless\Deterministic}} \mathbb{E}_{\mathcal{G}_{A,A'},v'}^{\strat{\Verifier},\strat{\Refuter}}[\FMask]
    \leq 
     \inf_{\strat{\Refuter} \in \Strategies{\Refuter}^{\Memoryless\Deterministic}} \sup_{\strat{\Verifier} \in \Strategies{\Verifier}^{\Memoryless\Deterministic}} \mathbb{E}_{\mathcal{H}_{A,A'},v'}^{\strat{\Verifier},\strat{\Refuter}}[\FMask],
\end{equation}
\ref{eq:theo:determined:eq4} holds since $\post^{\mathcal{H}_{A,A'}}(v) \subseteq \post^{\MilestoneG_{A,A'}}(v)$ for $v \in V^{\mathcal{H}_{A,A'}}_\Verifier$ and
$\post^{\mathcal{H}_{A,A'}}(v) = \post^{\MilestoneG_{A,A'}}(v)$ for $v \in V^{\mathcal{H}_{A,A'}}_\Refuter$.  For proving \ref{eq:theo:determined:eq5}, fix a
fair strategy $\strat{\Refuter} \in \Strategies{\Refuter}^{\Memoryless \Deterministic}$, the optimal strategy for the Verifier in game $\mathcal{G}_{A,A'}$ is attained only
in probabilistic vertices that are vertices of $\couplings{\mu}{\mu'}$, thus the probabilistic vertices of $\mathcal{H}_{A,A'}$, thus 
$  \sup_{\strat{\Verifier} \in \Strategies{\Verifier}^{\Memoryless\Deterministic}} \mathbb{E}_{\mathcal{G}_{A,A'},v'}^{\strat{\Verifier},\strat{\Refuter}}[\FMask]
    \leq 
     \sup_{\strat{\Verifier} \in \Strategies{\Verifier}^{\Memoryless\Deterministic}} \mathbb{E}_{\mathcal{H}_{A,A'},v'}^{\strat{\Verifier},\strat{\Refuter}}[\FMask],$
for any fair and memoryless $\strat{\Refuter}$, \ref{eq:theo:determined:eq5} follows.

    Furthermore,  the value of game $\mathcal{H}_{A,A'}$ is given by  the greatest fixed point of the equations \cite{DBLP:journals/corr/abs-2112-09811}:
\begin{equation}\label{eq:theo:determined:bellman1}
    \Bellman(f)(v) =
    \begin{cases}
           \displaystyle \min \{\upperbound,  \chi_{\mathcal{M}}(\pr{1}{v}) + \sum_{v' \in \post(v)} \delta(v)(v')  f(v') \} & \text{ if } v \in V^{\mathcal{H}_{A,A'}}_\Probabilistic  \\
           \displaystyle \min \{\upperbound,  \max \{\chi_{\mathcal{M}}(\pr{1}{v})  +f(v') \mid v' \in \post(v) \} \} & \text{ if } v \in  V^{\mathcal{H}_{A,A'}}_\Verifier, \\
           \displaystyle \min \{\upperbound,  \min \{\chi_{\mathcal{M}}(\pr{1}{v})  + f(v') \mid v' \in \post(v) \} \} & \text{ if } v \in  V^{\mathcal{H}_{A,A'}}_\Refuter, \\
           \displaystyle 0 & \text{ if } v=\ErrorSt.
           \displaystyle 
    \end{cases}
\end{equation}
    for some $\upperbound \geq \inf_{\strat{\Refuter} \in \Strategies{\Refuter}^{\Memoryless}} \sup_{\strat{\Verifier} \in \Strategies{\Verifier}^{\Memoryless}} \mathbb{E}_{\mathcal{H}_{A,A'},v'}^{\strat{\Verifier},\strat{\Refuter}}[\FMask]$.
    That is, we have:
\begin{equation}\label{eq:theo:determined:eq6}
    \inf_{\strat{\Refuter} \in \Strategies{\Refuter}^{\Memoryless\Deterministic}} \sup_{\strat{\Verifier} \in \Strategies{\Verifier}^{\Memoryless\Deterministic}}  \mathbb{E}^{\strat{\Verifier},\strat{\Refuter}}_{\mathcal{H}_{A,A'}, v}[\FMask] 
    = \sup_{\strat{\Verifier \in \Strategies{\Verifier}^{\Memoryless\Deterministic}}} \inf_{\strat{\Refuter \in \Strategies{\Refuter}^{\Memoryless\Deterministic}}}    \mathbb{E}^{\strat{\Verifier},\strat{\Refuter}}_{\mathcal{H}_{A,A'}, v}[\FMask] 
    < \infty
\end{equation}
  Thus, because of \ref{eq:theo:determined:eq4}, \ref{eq:theo:determined:eq5} and \ref{eq:theo:determined:eq6} we have:
\[
 \inf_{\strat{\Refuter} \in \Strategies{\Refuter}^{\Memoryless}} \sup_{\strat{\Verifier} \in \Strategies{\Verifier}^{\Memoryless}}  \mathbb{E}^{\strat{\Verifier},\strat{\Refuter}}_{\MilestoneG_{A,A'}, v}[\FMask] 
    = \sup_{\strat{\Verifier \in \Strategies{\Verifier}^{\Memoryless}}} \inf_{\strat{\Refuter \in \Strategies{\Refuter}^{\Memoryless}}}    \mathbb{E}^{\strat{\Verifier},\strat{\Refuter}}_{\MilestoneG_{A,A'}, v}[\FMask] 
    < \infty.
\]
    This proves a part of the theorem.
    Now, consider the following functional over the symbolic game:
\[
\Bellman'(f)(v) =
    \begin{cases}
           \displaystyle \min \{ \upperbound, \max_{w \in \vertices{\couplings{\pr{3}{v}}{\pr{4}{v}}}} \{\sum_{v' \in \post(v)} w(\pr{0}{v'},\pr{2}{v'})  f(v') \} \}& \text{ if }v \in V^{\SymbG_{A,A'}}_\Probabilistic  \\
           \displaystyle \min \{ \upperbound, \max \{\chi_{\mathcal{M}}(\pr{1}{v})  +f(v') \mid v' \in \post(v) \} \} & \text{ if } v \in  V^{\SymbG_{A,A'}}_\Verifier, \\
           \displaystyle \min \{ \upperbound, \min \{\chi_{\mathcal{M}}(\pr{1}{v})  + f(v') \mid v' \in \post(v) \} \} & \text{ if } v \in  V^{\SymbG_{A,A'}}_\Refuter, \\
           \displaystyle 0 & \text{ if } v=\ErrorSt.
           \displaystyle 
    \end{cases}
\]
we will prove that this can be used to solve $\Bellman$. First, note that $\Bellman'$ is monotone,  it is defined over a complete lattice $[0,\upperbound]$ and 
it is Scott-complete. Thus, it has a greatest fixpoint.  Let $\nu \Bellman'$ the greatest fixpoint, of $\Bellman',$  we prove that 
$\nu \Bellman(v) = \nu \Bellman'((v[0],v[1],v[2],v[3],v[4]))$, for every $v \in V^{\mathcal{H}_{A,A'}}_\Verifier \cup V^{\mathcal{H}_{A,A'}}_\Refuter$.
For doing so, consider for each symbolic vertex the following mapping:
\begin{itemize}
    \item $\llbracket (s,\sigma,s',\mu,\mu',X) \rrbracket = (s,\sigma,s',\mu,\mu',\mhyphen,X)$, for $X \in \{\Refuter, \Verifier\}$,
    \item $\llbracket (s,\mhyphen, s',\mu,\mu',\Probabilistic) \rrbracket =(s, \mhyphen,s', \mu,\mu', w,\Probabilistic)$,  \\ where
              $w = \argmax_{w \in \vertices{\couplings{\mu}{\mu'}}} \{\sum_{v' \in \post(v)} w(\pr{0}{v'},\pr{2}{v'})  \nu \Bellman'(v') \}$
\end{itemize}
    Similarly, we can define a mapping from concrete vertices to symbolic ones as:
\begin{itemize}
    \item $\llparenthesis (s,\sigma,s',\mu,\mu',Y ,X) \rrparenthesis = (s,\sigma,s',\mu,\mu',X)$, for $X \in \{\Refuter, \Verifier$ and $Y \in \{\mhyphen \} \cup \vertices{\couplings{\mu,\mu'}}$.
\end{itemize}
    Now, we prove that $\alpha(v) = \nu \Bellman'(\llparenthesis v \rrparenthesis)$ is a fixpoint of $\Bellman$. We proceed by cases:
   
   If $v$ is a Refuter's vertex,  then:
\begin{align}
   \Bellman(\alpha)(v) & =  \min \{\upperbound,  \min \{\chi_{\mathcal{M}}(\pr{1}{v})  + \alpha(v') \mid v' \in \post(v) \} \}  \\
                                    & =  \min \{\upperbound,  \min \{\chi_{\mathcal{M}}(\pr{1}{v})  +\nu \Bellman'(\llparenthesis v'  \rrparenthesis) \mid v' \in \post(v) \} \} \\  
                                    & = \nu \Bellman'(\llparenthesis v  \rrparenthesis) \\
                                    & = \alpha(v)           
\end{align}
 The first line is by definition of $\Bellman$, the second line is obtained applying the definition of $\alpha$, the third line is due to suryectivity of  $\llparenthesis \rrparenthesis$ and
 definition of $\Bellman'$ and the fact that  $\nu \Bellman'(\llparenthesis v  \rrparenthesis)$ is a fixed point of $\Bellman'$.
  If $v$ is a Probabilistic vertex,  the proof similar:
\begin{align}
   \Bellman(\alpha)(v) & =  \min \{ \upperbound, \max_{w \in \couplings{\pr{3}{v}}{\pr{4}{v}}} \{\sum_{v' \in \post(v)} w(\pr{0}{v'},\pr{2}{v'})  \alpha(v') \} \} \\
                                    & =   \min \{ \upperbound, \max_{w \in \couplings{\pr{3}{v}}{\pr{4}{v}}} \{\sum_{v' \in \post(v)} w(\pr{0}{v'},\pr{2}{v'})  \nu \Bellman'(\llparenthesis v' \rrparenthesis) \} \} \\  
                                    & = \nu \Bellman'(\llparenthesis v  \rrparenthesis) \\
                                    & = \alpha(v)           
\end{align}

    If $v$ is a Verifier's vertex, then:
\begin{align}
   \Bellman(\alpha)(v) & =   \min \{ \upperbound, \max \{\chi_{\mathcal{M}}(\pr{1}{v})  + \alpha(v') \mid v' \in \post(v) \} \}  \\
                                    & =   \min \{ \upperbound, \max_{w \in \vertices{\couplings{\pr{3}{v}}{\pr{4}{v}}}} \{\sum_{v' \in \post(v)} w(\pr{0}{v'},\pr{2}{v'})  \nu \Bellman'(\llparenthesis v' \rrparenthesis) \} \} \\  
                                    & = \nu \Bellman'(\llparenthesis v  \rrparenthesis) \\
                                    & = \alpha(v)           
\end{align}
    Hence, $\alpha$ is a fixpoint of $\Bellman$. Furthermore,  we prove that it is greatest one. Assume for the sake of contradiction that there is some
    $\alpha'$ such that is a fixpoint of $\Bellman$ and $\alpha'(v) \geq  \alpha(v)$ for every $v$, and $\alpha'(v') >  \alpha(v')$ for some $v'$.  
    We can define $\beta(v) = \alpha'(\llbracket v \rrbracket)$, as above we can prove that it is a fixpoint of $\Bellman'$ and, furthermore, for every 
    symbolic vertex we have $\beta(v) =  \alpha'(\llbracket v \rrbracket) \geq \alpha(\llbracket v \rrbracket) = \nu \Bellman'(\llparenthesis \llbracket v \rrbracket \rrparenthesis)
    = \nu \Bellman'(v)$, and similarly we can prove that there is a $v'$ such that $\beta(v') >  \nu \Bellman'(v)$, which is a contradiction since $\nu \Bellman'$ is 
    the greatest fixpoint of $\Bellman'$.
\end{proof}\\

\noindent
\textbf{Proof of Theorem \ref{theo:decide-stopping}}
Given a masking game $\StochG_{A,A'}$ and its symbolic version $\SymbG_{A,A'}$,  we have that 
 $\StochG_{A,A'}$ is stopping under fairness  iff
$
  \InitVertex \in V^\SymbG \setminus {\SymbEFairpre}^*(V^\SymbG \setminus {\SymbAFairpre}^*(\ErrorSt)),
$
where $\InitVertex$ is the initial state of $\SymbG_{A,A'}$ and $V^\SymbG$ its sets of vertices.

\noindent
\begin{proof} 
    Consider the game $\mathcal{H}_{A,A'}$ as defined in the proof of Theorem \ref{theo:det-fairness}.  First, we prove that 
the game $\StochG_{A,A'}$ is almost-sure failing for fair Refuter's strategies iff $\mathcal{H}_{A,A'}$ is too almost-sure failing for fair Refuter's strategies.  
This is equivalent to prove that $\inf_{\strat{\Verifier}} \Prob{\strat{\Verifier}}{\strat{\Refuter}}_{\StochG_{A,A'}, \InitVertex}(\Diamond \ErrorSt)=1$ 
iff $\inf_{\strat{\Verifier}} \Prob{\strat{\Verifier}}{\strat{\Refuter}}_{\mathcal{H}_{A,A'}, \InitVertex}(\Diamond \ErrorSt)=1$ for every strategy memoryless  
and fair $\strat{\Refuter}$. Now, note we have:
\begin{align*}
     \Prob{\strat{\Verifier}}{\strat{\Refuter}}_{\mathcal{H}_{A,A'}, v}(\Diamond \ErrorSt) & =\min\{ \sum_{v' \post(v)} w(\pr{0}{v'}, \pr{2}{v'})\Prob{\strat{\Verifier}}{\strat{\Refuter}}_{\mathcal{H}_{A,A'}, v'}(\Diamond \ErrorSt) \mid w \in \couplings{\pr{3}{v}}{\pr{4}{v}} \}\\
         & = \min\{ \sum_{v' \post(v)} w(\pr{0}{v'}, \pr{2}{v'})\Prob{\strat{\Verifier}}{\strat{\Refuter}}_{\mathcal{H}_{A,A'}, v'}(\Diamond \ErrorSt) \mid w \in\vertices{\couplings{\pr{3}{v}}{\pr{4}{v}}} \}
\end{align*}
Thus, we have a deterministic and memoryless $\strat{\Verifier}'$ such that:
\[
\Prob{\strat{\Verifier}'}{\strat{\Refuter}}_{\mathcal{G}_{A,A'}, v}(\Diamond \ErrorSt) = \inf_{\strat{\Verifier} \in \Strategies{\Verifier}^{\Memoryless}} \Prob{\strat{\Verifier}}{\strat{\Refuter}}_{\StochG_{A,A'}, \InitVertex}(\Diamond \ErrorSt),
\]
this strategy only selects probabilistic vertices in $V^{\mathcal{H}}_{\Probabilistic}$, and therefore, $\strat{\Verifier}'$ is a strategy
in $\mathcal{H}_{A,A'}$. Noting that  the Markov chains $\StochG^{\strat{\Verifier}',\strat{\Refuter}}$,$ \mathcal{H}^{\strat{\Verifier}',\strat{\Refuter}}$
are the same for every strategy $\strat{\Refuter}$, we have: 
$\inf_{\strat{\Verifier}} \Prob{\strat{\Verifier}}{\strat{\Refuter}}_{\StochG_{A,A'}, v^{\mathcal{H}}_0}(\Diamond \ErrorSt)
= \inf_{\strat{\Verifier}} \Prob{\strat{\Verifier}}{\strat{\Refuter}}_{\mathcal{H}_{A,A'}, \InitVertex}(\Diamond \ErrorSt)$.

    Now, we prove that we can check wether the game $\mathcal{H}$ is almost-sure failing under fairness using the symbolic game. We define the following sets 
over this game:
\begin{align*}
  \EFairpre(C) = {}&\{ v \in V^{\mathcal{H}} \mid\exists v' \in C : (v,v') \in E^\mathcal{H}_{A,A'} \} \\
  \AFairpre(C) = {}&\{ v \in V^{\mathcal{H}}_\Probabilistic \mid \delta(v,C)>0\} \\
                       & \cup \{ v \in  V^{\mathcal{H}}_\Verifier \mid \forall v' {\in} V^{\mathcal{H}} : (v,v') \in E^{\mathcal{H}}_{A,A'} \Rightarrow v' {\in} C \} \\
                     & \cup \{v \in V^{\mathcal{H}}_\Refuter \mid \exists v'{\in} V^{\mathcal{H}} : (v,v') \in E^\mathcal{H}_{A,A'} \} 
\end{align*}

As proven in \cite{DBLP:journals/corr/abs-2112-09811} (Theorem 3) we have that: 
 $\Prob{\strat{\Verifier}}{\strat{\Refuter}}_{\mathcal{H},v}(\Diamond \ErrorSt) = 1$ for every
 $\strat{\Verifier} \in \Strategies{\Verifier}$ and fair $\strat{\Refuter} \in \Strategies{\Refuter}$
  iff $v \in V\setminus \EFairpre^*(V \setminus \AFairpre^*(\ErrorSt))$.
    
Now, we prove that for all $v \in V^\mathcal{H}$ and $v' \in V^\SymbG$ 
such that $\pr{i}{v} = \pr{i}{v'}$ for $0 \leq i \leq 4$, we have:
$v \in  \AFairpre^n(\ErrorSt)$ iff $v' \in {\SymbAFairpre}^n(\ErrorSt)$, for every $n$. 
The proof is by induction on $n$. The base case is direct. The inductive
cases for Refuter's  nodes and Probabilistic nodes are also direct (as the predecessors in both games are the same  for those vertices module dummy notation). For Verifier's nodes we proceed as follows. If $v \in  \AFairpre^n(\ErrorSt)$, then
for all $t \in \post(v)$ we have $\delta(t)( \AFairpre^{n-2}(\ErrorSt))>0$, that is, $\Eq(t')({\SymbEFairpre}^{n-2}(\ErrorSt))$ has no solutions (with $t'$ satisfying $\pr{i}{t} = \pr{i}{t'}$ for $0 \leq i \leq n$ ). 
Thus, $t' \in {\SymbAFairpre}^{n-1}(\ErrorSt)$, and since $t'$ is the unique successor of $v'$ we have that $v' \in {\SymbAFairpre}^n(\ErrorSt)$, the other direction is similar.

	Now, given sets $C \subseteq V^\mathcal{H}$ and $C^* \subseteq V^\StochG$	 with $C^* = \{ v \mid \exists v' \in C : \forall 0 \leq i \leq 4 : \pr{i}{v} = \pr{i}{v'}\}$.  We prove
that, for every $u \in V^\mathcal{H}$ and $u' \in V^\SymbG$ such that $\pr{i}{u} = \pr{i}{u'}$ for every $0 \leq i \leq 4$, we have $u \in \EFairpre^n(C)$ iff $u' \in \EFairpre^n(C^*)$. 
The proof is by induction on $n$, for $n=0$ it is direct. For $n>0$, assume $w \in \EFairpre^n(C)$,  for Refuter's (or Verifier's) vertices the proof is direct, since in both games they have the same successors (up to removal of dummy notation).  If $u$ is a Probabilistic vertex and $u \in \EFairpre^n(C)$, then there is $t \in C$ such that 
$t \in \post(u)$, that is, $\pr{5}{u}(\pr{3}{t})(\pr{4}{t}) > 0$. But then  $\pr{0}{t} \in \support{\pr{3}{u}}$ and $\pr{2}{t} \in \support{\pr{4}{u}}$ which
implies that $u' \in {\SymbEFairpre}^n(C)$. The other direction is similar, but noting that if we have $\pr{0}{t} \in \support{\pr{3}{u}}$ and $\pr{2}{t} \in \support{\pr{4}{u}}$, then
we can construct a coupling relating distributions $\pr{3}{u}$ and $\pr{4}{u}$.

    Thus, we have that $\InitVertex \in V\setminus \EFairpre^*(V \setminus \AFairpre^*(\ErrorSt))$ iff   $\InitVertex \in V^\SymbG \setminus {\SymbEFairpre}^*(V^\SymbG \setminus {\SymbAFairpre}^*(\ErrorSt))$, from there the theorem follows.
\end{proof}

\end{document}